\documentclass[a4paper]{article}
\usepackage{amsmath,amssymb,amsthm,amsfonts,graphicx,fullpage}

\usepackage[usenames,dvipsnames]{xcolor}

\usepackage[misc,geometry]{ifsym}

\usepackage{natbib}

\usepackage{stmaryrd}

\usepackage{url}
\usepackage{wasysym}

\usepackage{hyperref}

\newcommand{\cK}{{\cal K}}

\newcommand{\Rset}{\mathbb{R}}
\newcommand{\Nset}{\mathbb{N}}

\newcommand{\inte}{{\mathrm {\tt int}}}

\newcommand{\comment}[1]{}
\renewcommand{\t}{^{\mbox{\tiny\sf T}}}

\newtheorem{theo}{Theorem}
\newtheorem*{theo*}{Theorem}
\newtheorem{defi}[theo]{Definition}
\newtheorem{lemm}[theo]{Lemma}

\newtheorem*{prop*}{Proposition}

\theoremstyle{definition}
\newtheorem{rema}[theo]{Remark}
\newtheorem*{rema*}{Remark}
\newtheorem*{remas*}{Remarks}

\begin{document}

\title{Ensuring successful introduction of {\em Wolbachia} in natural 
populations of {\em Aedes aegypti} by means of feedback control}




\author{Pierre-Alexandre Bliman\thanks{Escola de Matem\'atica Aplicada, Funda\c c\~ ao Getulio Vargas, Praia de Botafogo 190,   22250-900 Rio de Janeiro - RJ, Brazil and Inria, Domaine de Voluceau, Rocquencourt BP105, 78153 Le Chesnay cedex, France,
\href{mailto:pierre-alexandre.bliman@inria.fr}{pierre-alexandre.bliman@inria.fr}} \and M.\ Soledad Aronna\thanks{IMPA, Estrada Dona Castorina 110, 22460-320 Rio de Janeiro - RJ, Brazil and Escola de Matem\'atica Aplicada, Funda\c c\~ ao Getulio Vargas, Praia de Botafogo 190,   22250-900 Rio de Janeiro - RJ, Brazil, \href{mailto:aronna@impa.br}{aronna@impa.br}} \and Fl\'avio C.\ Coelho\thanks{Escola de Matem\'atica Aplicada, Funda\c c\~ ao Getulio Vargas, Praia de Botafogo 190,   22250-900 Rio de Janeiro - RJ, Brazil, \href{mailto:fccoelho@fgv.br}{fccoelho@fgv.br}} \and Moacyr A.H.B.\ da Silva\thanks{Escola de Matem\'atica Aplicada, Funda\c c\~ ao Getulio Vargas, Praia de Botafogo 190,   22250-900 Rio de Janeiro - RJ, Brazil, \href{mailto:moacyr@fgv.br}{moacyr@fgv.br}}}

\date{\today}


\maketitle

\begin{abstract}
The control of the spread of dengue fever by introduction of the intracellular parasitic bacterium {\em Wolbachia} in populations of the vector {\em Aedes aegypti}, is presently one of the most promising tools for eliminating dengue, in the absence of an efficient vaccine. 
The success of this operation requires locally careful planning to determine the adequate number of individuals carrying the {\em Wolbachia} parasite that need to be introduced into the natural population.
The  introduced mosquitoes are expected to eventually replace the {\em Wolbachia}-free 
population and guarantee permanent protection against the transmission of 
dengue to human.

In this study, we propose and analyze a model describing the fundamental 
aspects of the competition between mosquitoes carrying {\em Wolbachia} and mosquitoes free of the parasite.
We then use feedback control techniques to devise an introduction protocol which is proved to guarantee that the population converges to a stable equilibrium where the totality of mosquitoes carry {\em Wolbachia}.

\vspace{.15cm}
\noindent
{\bf Keywords} Wolbachia; Global stabilization; Monotone systems; Input-output characteristic

\vspace{.15cm}
\noindent
{\bf Mathematics Subject Classification (2000)} Primary 92D30; Secondary 34C12, 93D15, 93D25

\end{abstract}

\tableofcontents


\section{Introduction}

\subsection{Arboviroses and vector control}

Arboviruses (arthropod borne viruses) are viruses transmitted to humans by arthropod, such as the mosquito.
They are pathogens of many and important diseases, putting at risk considerable portions of the human population, and infecting 
 millions of people every year.
Mosquitoes ({\em Culicidae} family of the insects) are a huge public 
health concern as they are vectors of many arboviroses such as yellow fever, 
dengue and chikungunya.

The control of these diseases can be achieved by acting on the population of 
mosquitoes, and in absence of vaccine or curative treatment, it is essentially the only feasible way.
Application of insecticides for both adults and larvae and mechanical removal 
of breeding sites are the most popular strategies to control the population of mosquitoes.
The intensive use of insecticides, however, has negative 
impacts for humans, animals and the environment.
Besides, the gradual increase of mosquito resistance to insecticides usually leads to partial or complete decrease of the efficiency of this strategy in the long run 
\citep{MacieldeFreitas2014,montella2007insecticide}.
In addition to chemical control and mechanical removal of the breeding sites, 
alternative or supplementary vector control strategies have been proposed and 
implemented, such as the release of transgenic or sterile mosquitoes 
\citep{Alphey2010,Alphey2014}.
Notice that an intrinsic weakness of the techniques listed 
above lies paradoxically in the fact that they aim at the local eradication of 
the vector, whose disappearance offers no protection against subsequent 
reinvasions.

Recently the release of {\em Aedes aegypti} mosquitoes infected by the bacterium
{\em Wolbachia} has been proposed as a promising strategy to control
dengue and chikungunya, due to the fact that this bacterium limits severely the 
vectorial competence of {\em Aedes aegypti}.
The international program Eliminate Dengue \citep{EliminateDengue} is currently testing in the field this strategy, in 
several locations around the world: Australia, Indonesia, Vietnam, Colombia and 
Brazil.
The release of infected mosquitoes with {\em Wolbachia} has the advantages of being 
safe for humans and the environment and inexpensive when compared to other 
control strategies \citep{popovici2010assessing}.

\subsection{{\em Wolbachia} sp.\ as a biological control tool}

{\em Wolbachia} sp.\ is a {\em genus} of bacteria that is a common 
intracellular parasite of many species of arthropods.
It is often found in anthropophilic mosquitoes such as {\em Aedes albopictus} or 
{\em Culex quinquefasciatus} but there is no report of {\em Aedes aegypti} naturally 
infected by this bacterium \citep{rasgon2004initial}. 

There is evidence that the spread of certain strains of {\em Wolbachia} in 
populations of {\em Aedes aegypti} drastically reduces the vector 
competence of the mosquito for dengue and other diseases \citep{Moreira2009,Blagrove:2012aa,Blagrove:2013aa}.
Some strains of {\em Wolbachia} reduce the lifespan of the mosquito, consequently limiting the proportion of surviving mosquitoes at the completion of the incubation period.
More importantly, {\em Wolbachia} appears to decrease the virulence of 
the dengue infection in the mosquitoes, increasing the incubation period or 
blocking the virus, which also reduces the overall vector 
competence.

The infestation of natural {\em Aedes aegypti} populations by {\em Wolbachia}-conta\-mi\-na\-ted strains can be achieved by releasing in the field a large number 
of {\em Wolbachia}-infected mosquitoes bred in laboratory.
Experiments have been conducted successfully in Australia \citep{Hoffmann2011}, 
Vietnam, Indonesia and are currently being applied in Brazil, within the 
Eliminate Dengue Program.
In these experiments, the introduction of a number of {\em Wolbachia}-infected mosquitoes in the population triggered a 
{\em Wolbachia} outbreak whose outcome was the fixation of the bacteria in the 
population, with more than 90\% of prevalence.
The effects of this fixation of {\em Wolbachia} on the dynamics of dengue in the 
field 
is currently under investigation, but preliminary results are encouraging 
\citep{Frentiu2014}.
If tractable in practice, this method has certainly the advantage 
of offering certain resilience to subsequent invasion of {\em Wolbachia}-free mosquitoes.

Several mathematical models of the dynamics of invasion of {\em Wolbachia}
in a population of mosquitoes have been proposed, with distinct objectives.
For example, \cite{turelli2010cytoplasmic} describes a simple model with a single differential
equation, sufficient to reveal the bistable nature of the {\em Wolbachia} dynamics.
Models for spatial dispersion are analyzed in \citep{barton2011spatial} and \citep{hancock2012modelling}.
In \citep{Hughes:2013aa,ndii2015modelling}, models are presented that assess the effect of the {\em Wolbachia} in dengue dynamics.
\cite{koiller:hal-00939411} describes a data driven model suitable to estimate accurately
some biological parameters by fitting the model with field and lab data.
The model to be presented here is a simplified version of the latter.

\subsection{Description of the problem}

A key question about the introduction of {\em Wolbachia} in 
wild mosquitoes concerns the effective strategies of release of infected 
mosquitoes in the field that can be applied with limited cost to reach the desired state of 100\% of {\em Wolbachia}-carrying mosquitoes.
In this paper we propose and analyze a simple model of the dynamics of {\em Wolbachia}, that allows to investigate these strategies.

The main features of the natural dynamics of {\em Wolbachia} that have to be present 
in the model are the vertical transmission and the peculiar interference on the reproductive outcomes induced by cytoplasmic incompatibility \citep{InfluencialPassengers98}.
The transmission of {\em Wolbachia} occurs only vertically (i.e.\ from mother to the 
offspring), there is no transmission by contact.
Cytoplasmic incompatibility (CI) occurs when a female uninfected by {\em Wolbachia} is inseminated by an infected male, a mating that leads to sterile eggs.
It provides a benefit to infected females against uninfected ones and therefore facilitates the {\em Wolbachia} spread.
The strains {\em w}Mel and {\em w}MelPop of {\em Wolbachia} that are being used in the field experiments with {\em Aedes aegypti} induce almost total CI \citep{walker2011wmel}.
Table \ref{tab:incompatibilidade-citoplasm=0000E1tica} schematizes the results of the mating of infected and uninfected mosquitoes when the CI is 100\%.
The model that we propose below captures all of these features and is simple enough to allow for a fairly complete analysis.

\begin{table}
\begin{center}
\begin{tabular}{|c||c|c|}
\hline 
  & Uninfected $\male$ & Infected $\male$\\
\hline 
\hline 
Uninfected $\female$ & \bf Uninfected & \bf \color{red}Sterile eggs \\
\hline 
Infected $\female$ & \bf Infected  & \bf Infected \\
\hline 
\end{tabular}
\end{center}
\caption{\label{tab:incompatibilidade-citoplasm=0000E1tica}Cytoplasmic incompatibility and vertical transmission of {\em Wolbachia} bacteria.
The state of the offspring is indicated, depending on the parents status}
\label{ta2}
\end{table}

The corresponding system is shown to possess two unstable equilibria, which correspond to 
extinction of the two populations and a coexistence equilibrium; and two locally 
asymptotically stable equilibria, which correspond to {\em Wolbachia}-free and 
complete infestation equilibria.
The release strategy we propose here is based on techniques from Control theory --- the released quantity of infected mosquitoes being the 
{\em control input}.
Using the fact that measurements are achieved and available during the whole 
release process, we propose a simple {\em feedback control law} that uses this 
knowledge to compute the input value.
The key result of the paper (Theorem \ref{th12}) proves that this control law has the capacity to 
asymptotically settle the bacterium {\em from whatever initial conditions}, and 
in particular from the completely {\em Wolbachia}-free equilibrium.
A major advantage of feedback compared to {\em open-loop approaches} (when the 
release schedule is computed once for all before the beginning of the 
experiment), is its ability to cope with the uncertainties on the model dynamics 
(e.g.\ in the modeling of the life stages and the population structure), on the 
parameters (population size, mortality, reproductive rates, etc.), and on the 
conditions of the realization (in particular on the size of the population to be 
treated).

Up to our knowledge, the present paper constitutes the first attempt to use 
feedback approach for introduction of {\em Wolbachia} within a population of arthropods.
Notice that we treat here  only the case of the release of
{\em Wolbachia-positive larvae} and full information on the quantity of 
{\em Wolbachia-negative larvae}.
Yet, the same dynamical model offers the ability to study other configurations, 
both for the control and the observation, and the corresponding issues will be 
examined in future work.



The paper is organized as follows.
The simple model used in the sequel is introduced in Section \ref{se1}, and normalized.
The analysis of the uncontrolled model is made in Section \ref{se2}, showing the bistability announced above between {\em Wolbachia}-free equilibrium and full infestation.
A proportional control law is then proposed in Section \ref{se3}, and proved to lead to global stability of the full infestation equilibrium.
Simulations are provided in Section \ref{se4}.
Last, concluding remarks complete the text in Section \ref{se5}.

\paragraph{Notation} For $n\in \Nset,$ we let $\Rset^n$ denote the $n$-dimensional Euclidean real space, and with $\Rset^n_+$ we refer to the cone consisting of vectors in $\Rset^n$ with {\em nonnegative} components.
We write $\max\{a;b\}$ (resp.\ $\min\{a;b\}$) for the maximum (resp.\ minimum) of two real numbers $a,b$.

\section{A simple model of infestation by {\em Wolbachia}}
\label{se1}

The simplified compartment model we introduce includes two life stages:
a preliminary one, gathering the aquatic phases (egg and larva) where the mosquitoes are subject to space and food competition; and an adult one, representing all the posterior aerial phases (pupae, immature and mature adult).
Accordingly, we will denote $\mathbf L$ and $\mathbf A$ the corresponding state variables.
The uninfected and infected (by {\em Wolbachia}) populations will be distinguished by indexes $U$ and $W$ respectively, so we end up with a four state variables model, namely $\mathbf L_U, \mathbf L_W$ and $\mathbf A_U, \mathbf A_W$, that represent the numbers of uninfected, resp.\ infected, vectors in preliminary and adult phases.

We propose the following evolution model.
\begin{subequations}
\label{eq1}
\begin{gather}
\label{eq1a}
\dot {\mathbf L}_U
= \alpha_U \frac{\mathbf A_U}{\mathbf A_U+\mathbf A_W}\mathbf A_U -\nu \mathbf L_U - \mu (1+k (\mathbf L_W+\mathbf L_U))\mathbf L_U \\
\label{eq1b}
\dot {\mathbf A}_U = \nu \mathbf L_U -\mu_U \mathbf A_U\\
\label{eq1c}
\dot {\mathbf L}_W = \alpha_W \mathbf A_W -\nu \mathbf L_W - \mu (1+k (\mathbf L_W+\mathbf L_U))\mathbf L_W + \mathbf u \\
\label{eq1d}
\dot {\mathbf A}_W = \nu \mathbf L_W -\mu_W \mathbf A_W
\end{gather}
\end{subequations}
All the parameters are positive, their meaning is summarized in Table \ref{ta1}.
\begin{table}
\begin{center}
\begin{tabular}{|r|l|}
\hline
Notation & Meaning\\
\hline\hline
$\alpha_U, \alpha_W$ & Fecundity rates of uninfected and infected insects\\
$\nu$ & Rate of transfer from the preliminary to the adult stage\\
$\mu$ & Mortality rate of uninfected and infected insects in preliminary stage\\
$\mu k$ & Characteristic of the additional mortality rate in preliminary stage\\
$\mu_U, \mu_W$ & Mortality rates of uninfected and infected insects at adult stage\\
\hline
\end{tabular}
\caption{List of parameters of model \eqref{eq1}}
\label{ta1}
\end{center}
\end{table}

Most aspects of this compartment model are rather classical, we now comment on the most original modeling choices.
 The differences between the behaviors of the two populations lie in the different fecundity and mortality rates.
The mortality during the larva stage and the duration of the latter are considered unmodified by the disease.
Also, the (quadratic) competition term is assumed to act equally on both populations, with an effect proportional to the global number of larvae.

The first effect of {\em Wolbachia} is to diminish fertility and life duration, leading to reduced fitness for the infected mosquitoes.
This assumption will correspond to the choice of parameters made in \eqref{eq2} below.
The second effect of {\em Wolbachia}, namely the cytoplasmic incompatibility, modeled here as complete, forbids fecund mating between infected males and uninfected females.
This is rendered in \eqref{eq1a} by a recruitment term proportional altogether to the  {\em number} and to the {\em ratio} of uninfected adults.
Notice in this respect that the model apparently does {\em not} make distinction between males and females.
In fact, one could introduce in place of the two variables $\mathbf A_U, \mathbf A_W$, four variables $\mathbf F_U, \mathbf F_W, \mathbf M_U, \mathbf M_W$ representing the quantities of female and male adults.
The recruitment terms in \eqref{eq1c}, resp.\ \eqref{eq1a}, would then naturally be replaced by expressions proportional to $\mathbf F_W$, resp.\ $\frac{\mathbf M_U}{\mathbf M_U+\mathbf M_W}\mathbf F_U$.
However it is easy to see that the proportion between males and females remains constant  in this more complex system, as long as the mortality rates for healthy and infected insects are equal for the males and the females, and the passage from larvae to adults occurs with a constant sex ratio.
One can therefore use a single variable to take account of the uninfected adults, and a single one to take account of the infected adults, just as done in \eqref{eq1}.
In other terms, provided the sex ratio is constant and the mortality is sex-independent, a sexual model yields no more information than \eqref{eq1}.

As a last comment, notice the term $\mathbf u$ in equation \eqref{eq1c}.
The latter is an {\em input variable}, modeling the on-purpose introduction of infected larvae in the system in order to settle {\em Wolbachia}.

\subsection{Normalization and general assumption}

In order to reduce the number of parameters and to exhibit meaningful quantities, we now normalize model \eqref{eq1}.
Defining
\begin{gather*}
 L_\eta (t) := \frac{k\mu}{\nu+\mu} \mathbf L_\eta\left(
\frac{t}{\nu+\mu}
\right),\quad
 A_\eta (t) := \frac{k\mu}{\nu} \mathbf A_\eta\left(
\frac{t}{\nu+\mu}
\right),\qquad
\eta = U,W \\
u(t) := \frac{k\mu}{(\nu+\mu)^2} \mathbf u \left(
\frac{t}{\nu+\mu}
\right)
\end{gather*}
with the following choice of {\em dimensionless} parameters
\begin{equation}
\label{eq900}
\gamma_\eta := \frac{\mu_\eta}{\nu+\mu},\quad
{\cal R}_0^\eta := \frac{\nu\alpha_\eta}{(\nu+\mu)\mu_\eta},\qquad
\eta = U,W
\end{equation}
the following normalized model is deduced, that will be used in the remainder of the paper.
\begin{subequations}
\label{eq10}
\begin{gather}
\label{eq10a}
\dot L_U = \gamma_U{\cal R}_0^U \frac{ A_U}{ A_U+ A_W} A_U - (1+ L_W+ L_U) L_U \\
\label{eq10b}
\dot A_U =  L_U -\gamma_U  A_U\\
\label{eq10c}
\dot L_W = \gamma_W {\cal R}_0^W  A_W - (1+ L_W+ L_U) L_W + u \\
\label{eq10d}
\dot A_W = L_W -\gamma_W  A_W
\end{gather}
\end{subequations}
The state variable for system \eqref{eq10} will be denoted
\begin{equation*}
x := (L_U,A_U,L_W,A_W)\ ,
\end{equation*}
and for sake of simplicity, we write \eqref{eq10} as
\begin{equation}
\label{eq00}
\dot x = f(x)+ Bu,
\end{equation}
where $f$ and $B$ are defined as
\begin{equation*}
f(x) := \begin{pmatrix}
\gamma_U{\cal R}_0^U \frac{ A_U}{ A_U+ A_W} A_U - (1+ L_W+ L_U) L_U \\
 L_U -\gamma_U  A_U\\
\gamma_W {\cal R}_0^W  A_W - (1+ L_W+ L_U) L_W \\
L_W -\gamma_W  A_W
\end{pmatrix},\qquad
B := \begin{pmatrix}
0 \\ 0 \\ 1 \\ 0
\end{pmatrix}\ .
\end{equation*}

We assume in all the sequel
\begin{equation}
\label{eq2}
{\cal R}_0^U > {\cal R}_0^W >1\ .
\end{equation}
Assumption \eqref{eq2}, which will be valid in the remainder of the paper, ensures the sustainability of each of the two isolated populations, with an even better sustainability for the non-infected one.
 See Theorem \ref{th3} below for more details.

\section{Analysis of the uncontrolled system}
\label{se2}

The uncontrolled system is obtained by taking zero input $u$, that is:

\begin{equation}
\label{eq6}
\dot x = f(x)
\end{equation}

\subsection{Well-posedness, positivity and boundedness}

One first shows the well-posedness of the Cauchy problem related to equation \eqref{eq6} for nonnegative initial conditions.

\begin{theo}
\label{th1}
For any initial value in $\Rset_+^4$, there exists a unique solution to the initial value problem associated to system \eqref{eq6}. 
The latter is defined on $[0,+\infty)$, depends continuously on the initial conditions and takes on values in $\Rset_+^4$.
Moreover, it is uniformly ultimately bounded.
\end{theo}

\begin{rema}
The previous result shows that system \eqref{eq6} is {\em positive.}
Therefore, when talking about ``trajectories", we will always mean trajectories with initial values in $\Rset_+^4$.
The same shortcut will be used for all positive systems considered later.
\end{rema}

Let us introduce the following definition of an {\em order induced by a cone}, that will be instrumental in proving Theorem \ref{th1}.

\begin{defi}
\label{de1}
Let $n\in\Nset$ and $\cK\subset\Rset^n$ be a closed convex cone with nonempty interior.
We use $\geq_\cK$ to denote the {\em order induced by  $\cK$,} that is: $\forall x,x'\in\Rset^n$,
\begin{subequations}
\begin{equation*}
x\geq_\cK x'\quad \Leftrightarrow\quad  x-x'\in \cK\ .
\end{equation*}
Similarly, one defines
\begin{gather*}
x>_\cK x' \Leftrightarrow  x-x'\in \cK \text{ and } x\neq x'\\
x\gg_\cK x' \Leftrightarrow  x-x'\in \inte\ \cK
\end{gather*}
\end{subequations}
\end{defi}
As usual, one will write $x\leq_\cK x'$ to mean $x'\geq_\cK x$.

\begin{proof}[Proof of Theorem \ref{th1}]
Function $f$ in \eqref{eq6} is clearly well-defined and continuous in $\Rset_+^4$, except in points where $A_U=A_W=0$.
Due to the fact that $0\leq \frac{A_U}{A_U+A_W}\leq 1$, the quantity $\frac{A_U}{A_U+A_W}A_U$ tends towards zero when one approaches such points, and $f(x)$ can thus be defined by continuity when $A_U=A_W=0$.
 In addition, the right-hand side is clearly locally Lipschitz in $\Rset_+^4$, and classical result ensures the {\em local} well-posedness of the initial value problem, as long as the trajectory does not leave this set.
 
 The invariance property of the set $\Rset_+^4$ is verified due to the fact that
 \begin{equation*}
 \forall x\in\Rset_+^4,\,\forall i\in\{1,2,3,4\}: \qquad x_i =0\ \Rightarrow f_i(x)\geq 0\ .
 \end{equation*}
 
Let us now show that, for any initial condition in $\Rset_+^4,$ the associated trajectory remains bounded for all $t\geq 0.$
With this aim, let us define
\begin{subequations}
\label{eq800}
\begin{gather}
L := L_U+L_W,\qquad A := A_U+A_W,\\
\gamma := \min\{\gamma_U;\gamma_W\}>0,\qquad
{\cal R}_0 := \frac{\max\{\gamma_U{\cal R}_0^U;\gamma_W{\cal R}_0^W\}}{\gamma}\ .
\end{gather}
\end{subequations}
Notice that, in view of hypothesis \eqref{eq2},
\begin{equation}
\label{eq77}
{\cal R}_0 >1\ .
\end{equation}

It turns out that
\begin{subequations}
\begin{eqnarray*}
\dot L
& = &
\nonumber
\gamma_U{\cal R}_0^U\frac{A_U}{A_U+A_W}A_U+\gamma_W{\cal R}_0^WA_W
-(1+L_W+L_U)(L_U+L_W)\\
& \leq &
\nonumber
\left(
\gamma_U{\cal R}_0^UA_U+\gamma_W{\cal R}_0^WA_W
\right) - (1+L)L\\
& \leq &
\gamma{\cal R}_0A-(1+L)L
\end{eqnarray*}
and
\begin{equation*}
\dot A \leq L - \gamma A
\end{equation*}
\end{subequations}

Now, the auxiliary system
\begin{equation}
\label{eq30}
\dot L' = \gamma{\cal R}_0A'-(1+L')L',\qquad \dot A' = L' - \gamma A'
\end{equation}
is evidently {\em cooperative} (see \cite{Hirsch:1988aa}) for the canonic order induced by the cone $\Rset_+^2$.
One may thus use Kamke's theorem, see e.g.\ \cite[Theorem 10, p.\ 29]{Coppel:1965aa} or \cite{Smith:1995aa}, and compare the solutions of \eqref{eq6} (with $L$ and $A$ defined by \eqref{eq800}) and \eqref{eq30}.
One deduces 
$$
L(t)\leq L'(t),\quad A(t)\leq A'(t),\qquad \text{for all } t\geq 0,
$$
whenever the solutions are considered with the same initial conditions.

It may be shown without difficulty that system \eqref{eq30} possesses exactly two equilibria, namely
\begin{equation*}
x_*:=(L_*,A_*):=(0,0)\quad\text{ and }\quad x^*:=(L^*,A^*):=\left({\cal R}_0-1,\frac{1}{\gamma}{\cal R}_0\right)\ .
\end{equation*}
Due to \eqref{eq77}, linearization around each point shows that $x_*$ is locally unstable, while $x^*$ is locally asymptotically stable (LAS).
On the other hand, notice that $x_*\leq_{\Rset_+^2} x^*$. 
Using the local stability information, application of \cite[Theorem 10.3]{Hirsch:1988aa} then shows that the stability of $x^*$ is {\em global} in the topological interior of $\Rset_+^2$, and that this point is in fact attractive for any initial point distinct from $ x_*=(0,0)$.
In any case, all solutions of \eqref{eq30} converge to the order interval
\begin{equation*}
\llbracket x_*;x^*\rrbracket_{\Rset_+^2}
:= \left\{
x'\in\Rset_+^2\ :\ L_*\leq L' \leq L^*,\ A_*\leq A' \leq A^*
\right\}\ .
\end{equation*}

Coming back to the solutions of \eqref{eq6}, the comparison method mentioned above now allows to deduce the same property for $L,A$ defined in \eqref{eq800}.
Using finally the fact, proved in Theorem \ref{th1}, that the trajectories remain in $\Rset_+^4$, the same bounds apply componentwise to $L_U, L_W$ and $A_U, A_W$ respectively.
In particular, all solutions of \eqref{eq6} are uniformly ultimately bounded, and this yields global existence of  solutions, and hence the proof of Theorem \ref{th1}.
\end{proof}

\subsection{Monotonicity}

One shows here that the uncontrolled system \eqref{eq6} is {\em monotone.}
For sake of completeness, we recall here the definition of {\em monotone} and {\em strongly order-preserving} semiflows defined on a topological space $X$ partially ordered by an order relation $\leq_\cK$ generated by a cone $\cK$ (see \cite{Smith:1995aa}).

\begin{defi}
\label{defi2}
The semiflow $\Phi$ is called {\em monotone} if
\begin{equation*}
\Phi_t(x) \leq_\cK \Phi_t(x') \qquad \text{ whenever } x\leq_\cK x' \text{ and } t\geq 0\ .
\end{equation*}
The semiflow $\Phi$ is called {\em strongly order-preserving} if $\Phi$ is monotone and, whenever $x<_\cK x'$, there exist open subsets $\Omega,\Omega'$ of $X$ with $x\in \Omega$, $x'\in \Omega'$, and $t>0$ such that
\begin{equation*}
\Phi_t(\Omega) \leq_\cK \Phi_t(\Omega')\ ,
\end{equation*}
this meaning $z\leq_\cK z',$ for all $z\in\Phi_t(\Omega)$, $z' \in \Phi_t(\Omega')$.
The semiflow $\Phi$ is called {\em strongly monotone} if $\Phi$ is monotone and
\begin{equation*}
\Phi_t(x) \ll_\cK \Phi_t(x') \qquad \text{ whenever } x<_\cK x' \text{ and } t> 0\ .
\end{equation*}
A dynamical system is said to have one of the properties above if its associated semiflow does.
\end{defi}

We now examine system \eqref{eq6} at the light of these properties.

\begin{theo}
\label{th55}
System \eqref{eq6} is strongly order-preserving in $\Rset_+^4$ for the order induced by the cone 
\begin{equation}
\label{coneK}
\cK:=\Rset_-\times\Rset_-\times\Rset_+\times\Rset_+,
\end{equation}
 (that is such that: $x\geq_\cK x' \Leftrightarrow x_i\leq x'_i,\ i=1,2 \text{ and } x_i\geq x'_i,\ i=3,4$).
\end{theo}

System \eqref{eq6} is therefore monotone in $\Rset_+^4$, but not strongly monotone, due to the fact that the trajectories departing inside the sets $\{x\in\Rset_+^4\ :\ L_U =0, A_U=0\}$ and $\{x\in\Rset_+^4\ :\ L_W =0, A_W=0\}$ remain in these sets and, consequently, do not verify strict ordering property for the two null components.

Before proving Theorem \ref{th55}, we summarize in the following result the behavior of the trajectories in relation with some parts of the boundaries.
\begin{lemm}
\label{le7}
Let $x_0\in\Rset_+^4$.
Then exactly one of the four following properties is verified by the trajectories departing from $x_0$ at $t=0$.
\begin{itemize}
\item[$\star$]
$x\equiv x_{0,0}$ (that is, $x(t)=x_{0,0}$, $\forall t\geq 0$).
\item[$\star$]
$A_W\equiv 0$ and $A_U(t)>0$, $\forall t>0$.
\item[$\star$]
$A_U\equiv 0$ and $A_W(t)>0$, $\forall t>0$.
\item[$\star$]
$A_W(t)>0$ and $A_U(t)>0$, $\forall t>0$.
\end{itemize}
\end{lemm}

\begin{proof}[Proof of Lemma \ref{le7}]
Clearly, one sees from \eqref{eq10b} (resp. \eqref{eq10d}) that $A_U\equiv 0$ (resp.\ $A_W\equiv 0$) if and only if $A_U(0)=0$ and $L_U\equiv 0$ (resp.\ $A_W(0)=0$ and $L_W\equiv 0$).
Therefore, if $(A_U(0),L_U(0))\neq (0,0)$ (resp.\ $(A_W(0),L_W(0))\neq (0,0)$), then $A_U(t)>0$ (resp.\ $A_W(t)>0$) for all $t>0$.
This proves Lemma \ref{le7}.
\end{proof}

\begin{proof}[Proof of Theorem \ref{th55}]

We now introduce the gradient of $f$.
At each point $x= (L_U,A_U,L_W,A_W)\in\Rset_+^4$ such that $A_U+A_W>0$, $\nabla f(x)$ is equal to
\begin{equation}
\label{eq50}
\begin{pmatrix}
-1-2L_U-L_W & \gamma_U{\cal R}_0^U \left(
1- \frac{A_W^2}{(A_U+A_W)^2}
\right) & -L_U & -\gamma_U{\cal R}_0^U \frac{A_U^2}{(A_U+A_W)^2}\\
1 & -\gamma_U & 0 & 0\\
-L_W & 0 & -1-L_U-2L_W & \gamma_W{\cal R}_0^W\\
0 & 0 & 1 & -\gamma_W
\end{pmatrix}
\end{equation}
Notice that, as a corollary of Lemma \ref{le7}, either $x\equiv x_{0,0}$, or $A_U(t)+A_W(t)>0$ for all $t>0$.
Therefore, the gradient can be computed at any point of a trajectory, except if the latter is reduced to $x_{0,0}$.

For any $x\in\Rset_+^4$, one verifies easily that
\begin{subequations}
\begin{gather*}
\forall (i,j)\in\{1,2\}\times\{3,4\},\qquad
\frac{\partial f_i}{\partial x_j}(x) \leq 0,\
\frac{\partial f_j}{\partial x_i}(x) \leq 0,\\
\forall (i,j)\in\{1,2\}^2\cup\{3,4\}^2,\ i\neq j,\qquad
\frac{\partial f_i}{\partial x_j}(x) \geq 0\ .
\end{gather*}
\end{subequations}
Hence the system is monotone.

Moreover, except when $A_U=0$ or $A_W=0$, the Jacobian matrix in \eqref{eq50} is irreducible, and the semiflow related to system \eqref{eq6} is therefore strongly monotone therein.
On the other hand, trajectories confined to one of the sets $\{x\in\Rset_+^4\ :\ L_U =0, A_U=0\}$ and $\{x\in\Rset_+^4\ :\ L_W =0, A_W=0\}$, also verify strong monotonicity, {\em for the order relation restricted to the two non-identically zero components}.
These two remarks, together with Lemma \ref{le7}, show that overall the strongly order-preserving property is verified.
 This completes the demonstration of Theorem \ref{th55}.
\end{proof}

\subsection{Equilibrium points and stability}

The next result describes the situation of the equilibrium points and their stability.
Recall that the cone $\cK$ used to order the state space has been defined in \eqref{coneK} (in  Theorem \ref{th55}).

\begin{theo}
\label{th3}
System \eqref{eq6} possesses four equilibrium points, denoted $x_{0,0}$, $x_{U,0}$, $x_{0,W}$ and $x_{U,W}$ and corresponding respectively to zero population, disease-free state, complete infestation, and coexistence.
Moreover, the latter fulfill the following inequalities:
\begin{equation}
\label{eq776}
x_{U,0} \ll_\cK x_{U,W} \ll_\cK x_{0,W}\qquad\text{ and }\qquad
x_{U,0} \ll_\cK x_{0,0} \ll_\cK x_{0,W}\ .
\end{equation}
Last, the equilibrium points $x_{U,0}$ and $x_{0,W}$ are locally asymptotically stable (LAS), while the two other ones are unstable.
\end{theo}

The proof of Theorem \ref{th3} is decomposed in the following sections.

\subsubsection{Proof of Theorem \ref{th3} -- Computation and ordering of the equilibrium points}

One here computes the equilibrium points.
The latter verify
\begin{subequations}
\label{eq4}
\begin{gather}
\label{eq4a}
\gamma_U{\cal R}_0^U \frac{ A_U}{ A_U+ A_W} A_U - (1+ L_W+ L_U) L_U = 0\\
\label{eq4b}
\gamma_W {\cal R}_0^W  A_W - (1+ L_W+ L_U) L_W =0 \\
\label{eq4c}
L_U = \gamma_U  A_U,\qquad L_W =\gamma_W  A_W
\end{gather}
\end{subequations}

The point $x_{0,0} := (0,0,0,0)$ is clearly an equilibrium.
Let us look for an equilibrium $x_{U,0} := (L_U^*,A_U^*,0,0)$.
The quantities $L_U^*,A_U^*$ then have to verify
\begin{equation}
\label{eq11}
\gamma_U{\cal R}_0^U A_U^* - (1 + L_U^*)L_U^* =0,\qquad
 L_U^* = \gamma_U A_U^*\ .
\end{equation}
Dividing by $L_U^*\neq 0$ yields $1+ L_U^* = {\cal R}_0^U$.
One thus gets the unique solution of this form verifying
\begin{equation*}
L_U^* = {\cal R}_0^U-1,\qquad
A_U^* = \frac{{\cal R}_0^U-1}{\gamma_U}\ ,
\end{equation*}
which is positive due to hypothesis \eqref{eq2}.

Similarly, one now looks for an equilibrium defined as $x_{0,W} := (0,0,L_W^*,A_W^*)$.
The values of $L_W^*,A_W^*$ must verify
\begin{equation*}
\gamma_W {\cal R}_0^W A_W^* - (1+ L_W^*) L_W^*
 =0,\qquad
L_W^* = \gamma_W A_W^*\ .
\end{equation*}
This is identical to \eqref{eq11}, and as for the $x_{U,0}$ case, one gets a unique, positive, solution, namely
\begin{equation}
\label{eq21}
L_W^* = {\cal R}_0^W-1,\qquad
A_W^* = \frac{{\cal R}_0^W-1}{\gamma_W}\ .
\end{equation}

We show now that system \eqref{eq6} also admits a unique coexistence equilibrium with positive components $x_{U,W}= (L_U^{**}, A_U^{**}, L_W^{**}, A_W^{**})$.
Coming back to \eqref{eq4} and expressing the value of the factor common to the first and second identity leads to
\begin{eqnarray*}
1+L_U^{**}+L_W^{**}
& = &
\gamma_W{\cal R}_0^W\frac{A_W^{**}}{L_W^{**}}
= {\cal R}_0^W\\
& = &
\gamma_U{\cal R}_0^U\frac{A_U^{**}}{A_U^{**}+A_W^{**}}\frac{A_U^{**}}{L_U^{**}}
= {\cal R}_0^U\frac{A_U^{**}}{A_U^{**}+A_W^{**}}
\end{eqnarray*}
One thus deduces
\begin{subequations}
\begin{equation*}
\frac{A_U^{**}}{A_U^{**}+A_W^{**}} = \frac{{\cal R}_0^W}{{\cal R}_0^U}\ ,
\end{equation*}
and one can express all three remaining unknowns in function of $A_W^{**}$:
\begin{equation*}
L_W^{**} = \gamma_W A_W^{**},\quad
A_U^{**} = \frac{{\cal R}_0^W}{{\cal R}_0^U-{\cal R}_0^W}  A_W^{**},\quad
L_U^{**} = \gamma_U A_U^{**} = \gamma_U\frac{{\cal R}_0^W}{{\cal R}_0^U-{\cal R}_0^W}  A_W^{**}\ .
\end{equation*}
\end{subequations}
Using the value of $L_U^{**}$ and $L_W^{**}$ now yields the relation 
\begin{equation*}
{\cal R}_0^W -1
= L_U^{**}+L_W^{**}
= \gamma_W\left(
1+ \frac{\gamma_U}{\gamma_W}\frac{{\cal R}_0^W}{{\cal R}_0^U-{\cal R}_0^W}
\right)
A_W^{**}\ ,
\end{equation*}
which has a unique, positive, solution when \eqref{eq2} holds.
Setting for sake of simplicity
\begin{subequations}
\begin{equation*}
\delta := \frac{\gamma_U}{\gamma_W}\frac{{\cal R}_0^W}{{\cal R}_0^U-{\cal R}_0^W}\ ,
\end{equation*}
the fourth equilibrium is finally given by
\begin{gather*}
L_U^{**} = \frac{\delta}{1+\delta}({\cal R}_0^W-1),\qquad
A_U^{**} = \frac{\delta}{(1+\delta)\gamma_U}({\cal R}_0^W-1)\\
L_W^{**} = \frac{1}{1+\delta}({\cal R}_0^W-1),\qquad
A_W^{**} = \frac{1}{(1+\delta)\gamma_W}({\cal R}_0^W-1)
\end{gather*}
\end{subequations}
We have so far exhibited all the equilibrium points.

Notice that the last equilibrium can be expressed alternatively by use of the values of the equilibrium $x_{0,W}$:
\begin{subequations}
\label{eq770}
\begin{gather}
\label{eq770a}
L_U^{**} = \frac{\delta}{1+\delta}L_W^*,\qquad
A_U^{**} = \frac{\delta}{1+\delta}\frac{\gamma_W}{\gamma_U}A_W^*\\
\label{eq770b}
L_W^{**} = \frac{1}{1+\delta}L_W^*,\qquad
A_W^{**} = \frac{1}{1+\delta}A_W^*
\end{gather}
\end{subequations}
and this provides straightforward comparison result:
\begin{subequations}
\label{eq771}
\begin{equation}
\label{eq771a}
L_U^{**} < L_W^* <L_U^*
\qquad \text{ and }
\qquad
L_W^{**} < L_W^* <L_U^*
\end{equation}
and thus $A_\eta^{**} = \gamma_\eta L_\eta^{**}  < \gamma_\eta L_\eta^* = A_\eta^*$, for $\eta \in \{U,W\}$ and, therefore,
 \begin{equation}
\label{eq771b}
A_U^{**} <A_U^*\qquad\text{ and }\qquad A_W^{**} < A_W^*\ ,
\end{equation}
\end{subequations}
the second inequality being directly deduced from \eqref{eq770b}.
The relations \eqref{eq771} allow to establish the inequalities \eqref{eq776}.
 
\subsubsection{Proof of Theorem \ref{th3} -- Local stability analysis}
\label{se232}

The local stability analysis is conducted through the Jacobian matrices.
Recall that the gradient has been computed in \eqref{eq50}.

\paragraph{Stability of $x_{0,0}$.}

The value of $\nabla f$ at $x_{0,0}$ is not defined.
However, the trajectories issued from points in $\Rset_+^2\times\{0\}^2$, resp.\ $\{0\}^2\times\Rset_+^2$ clearly remain in the respective subspace, and the stability in these directions is controlled by the spectrum of the matrices
\begin{equation*}
\begin{pmatrix}
-1 & \gamma_U{\cal R}_0^U\\
1 & -\gamma_U
\end{pmatrix}
\qquad \text{ and } \qquad
\begin{pmatrix}
-1 & \gamma_W{\cal R}_0^W\\
1 & -\gamma_W
\end{pmatrix}\ .
\end{equation*}
Their eigenvalues are
\begin{equation*}
\frac{1}{2}\left(
-(1+\gamma_\eta) \pm \left(
(1+\gamma_\eta)^2+4\gamma ({\cal R}_0^\eta-1)
\right)^{1/2}
\right),\qquad \eta = U,W\ .
\end{equation*}
One of each pair is positive, due to condition \eqref{eq2}, and $x_{0,0}$ is thus {\em unstable}.

\paragraph{Stability of $x_{U,0}$.}

Using \eqref{eq50}, the gradient  $\nabla f(x_{U,0})$ of $f$ at $x_{U,0}$ is the upper block-triangular matrix
\begin{equation*}
\begin{pmatrix}
-1-2L_U^* & \gamma_U{\cal R}_0^U & -L_U^* & -\gamma_U{\cal R}_0^U\\
1 & -\gamma_U & 0 & 0\\
0 & 0 & -1-L_U^* & \gamma_W{\cal R}_0^W\\
0 & 0 & 1 & -\gamma_W
\end{pmatrix}
= \begin{pmatrix}
1-2{\cal R}_0^U & \gamma_U{\cal R}_0^U & 1-{\cal R}_0^U & -\gamma_U{\cal R}_0^U\\
1 & -\gamma_U & 0 & 0\\
0 & 0 & -{\cal R}_0^U & \gamma_W{\cal R}_0^W\\
0 & 0 & 1 & -\gamma_W
\end{pmatrix}\ .
\end{equation*}
Using the same arguments than the ones used to study $x_{0,0}$ to assess the stability of $2\times 2$ matrices, the two diagonal blocks are asymptotically stable if, respectively,
\begin{equation*}
\frac{{\cal R}_0^U}{2{\cal R}_0^U-1}<1
\qquad \text{ and } \qquad
{\cal R}_0^W<{\cal R}_0^U\ .
\end{equation*}
These conditions are realized when \eqref{eq2} holds.
In conclusion, the equilibrium $x_{U,0}$ is locally asymptotically stable.

\paragraph{Stability of $x_{0,W}$.}

The gradient $\nabla f(x_{0,W})$ of $f$ at $x_{0,W}$ is the lower block-triangular matrix
\begin{equation}
\label{eq26}
\begin{pmatrix}
-1-L_W^* & 0 & 0 & 0\\
1 & -\gamma_U & 0 & 0\\
-L_W^* & 0 & -1-2L_W^* & \gamma_W{\cal R}_0^W\\
0 & 0 & 1 & -\gamma_W
\end{pmatrix}
= \begin{pmatrix}
-{\cal R}_0^W & 0 & 0 & 0\\
1 & -\gamma_U & 0 & 0\\
1-{\cal R}_0^W & 0 & 1-2{\cal R}_0^W & \gamma_W{\cal R}_0^W\\
0 & 0 & 1 & -\gamma_W
\end{pmatrix}\ .
\end{equation}
The left-upper block is a Hurwitz matrix, while the asymptotic stability of the second one is equivalent to
\begin{equation*}
\label{eq5}
2{\cal R}_0^W-1 > {\cal R}_0^W\ ,
\end{equation*}
that is ${\cal R}_0^W>1$, which is true, due to hypothesis \eqref{eq2}.
The equilibrium $x_{0,W}$ is thus locally asymptotically stable.

\paragraph{Stability of $x_{U,W}$.}

The instability of $x_{U,W}$ can be proved by showing that the determinant of the Jacobian matrix  $\nabla f(x_{U,W})$ is negative, which, together with the fact that the state space has even dimension 4, establishes the existence of a positive real root to the  characteristic polynomial; and thus that the Jacobian is not a Hurwitz matrix.
This argument yields lengthy computations.

It is more appropriate to use here the monotonicity properties of system \eqref{eq6}, established in Theorem \ref{th55}.
As a matter of fact, bringing together the inequalities \eqref{eq776} (already proved in the end of the previous section, see \eqref{eq771}), the  asymptotical stability of $x_{U,0}$ and $x_{0,W}$ and the strongly order-preserving property of the system, \cite[Theorem 2.2]{Smith:1995aa} shows that the intermediary point $x_{U,W}$ cannot be stable.
This finally achieves the stability analysis, as well as the proof of Theorem \ref{th3}.

\subsection{Positively invariant sets and basins of attraction}

We exploit further in the following result the inequalities ordering the equilibrium points to have supplementary informations on some invariant sets.

\begin{theo}
\label{th5}
The order interval
\begin{equation*}
\llbracket x_{U,0};x_{0,W}\rrbracket_\cK
:=
\left\{
x \in\Rset^4 \ :\ x_{U,0} \leq_\cK x \leq_\cK x_{0,W} \subset\Rset_+^4
\right\}
\end{equation*}
is positively invariant for system \eqref{eq6}.
Moreover, the order interval $\llbracket x_{U,W};x_{0,W}\rrbracket_\cK$ (resp.\ $\llbracket x_{U,0};x_{U,W}\rrbracket_\cK$) is contained in the basin of attraction of $x_{0,W}$ (resp.\  $x_{0,U}$).
\end{theo}

\begin{proof}

The positive invariance properties are direct consequences of the monotonicity properties exhibited in Theorem \ref{th55}.
More precisely, endowing the state space with the ordering induced by the  cone  $\cK$ (see \eqref{coneK}), the autonomous system \eqref{eq10} induces a  monotone flow in $\Rset^4$, strongly monotone in $\Rset_+^4\setminus \left( \Rset_+^2\times\{0\}^2\cup \{0\}^2\times \Rset_+^2\right)$.
As the trajectories are bounded, the set of  initial points whose corresponding trajectories do not converge to one of the equilibria is of zero measure \cite[Theorem 7.8]{Hirsch:1988aa}.
Among the equilibria, only $x_{U,0}$ and $x_{0,W}$ are locally stable.

The same rationale applies for any trajectory with initial condition in the order interval $\llbracket x_{U,W};x_{0,W}\rrbracket_\cK\setminus\left( \Rset_+^2\times\{0\}^2\cup \{0\}^2\times \Rset_+^2\right)$, and the convergence (for almost every initial condition) in this interval can only occur towards $x_{0,W}$: the latter is therefore included in the basin of attraction.
The same argument applies for the other equilibrium $x_{U,0}$.
\end{proof}

\section{Analysis of the controlled system}
\label{se3}

\subsection{A class of static output-feedback control laws}

The following feedback law will be considered in the sequel:
\begin{equation}
\label{eq01}
u = KL_U
\end{equation}
for adequate (positive) values of the {\em scalar gain} $K$.
Writing
\begin{equation*}
e := \begin{pmatrix}
1 \\ 0 \\ 0 \\ 0
\end{pmatrix}
\end{equation*}
one obtains the closed-loop system:
\begin{equation}
\label{eq20}
\dot x = f(x) + KBe\t x\ ,
\end{equation}
or in developed form:
\begin{subequations}
\label{SysK}
\begin{gather}
\label{SysKa}
\dot L_U = \gamma_U{\cal R}_0^U \frac{ A_U}{ A_U+ A_W} A_U - (1+ L_W+ L_U) L_U \\
\label{SysKb}
\dot A_U =  L_U -\gamma_U  A_U\\
\label{SysKc}
\dot L_W = \gamma_W {\cal R}_0^W  A_W - (1+ L_W+ L_U) L_W + KL_U \\
\label{SysKd}
\dot A_W = L_W -\gamma_W  A_W
\end{gather}
\end{subequations}

The basic results gathered in the following theorem can be demonstrated by use of the same arguments than for  Theorem \ref{th1}.
The proof presents no difficulty and is left to the reader.

\begin{theo}
\label{th7}
For any initial value in $\Rset_+^4$, there exists a unique solution to the initial value problem associated to system \eqref{eq20}. 
The latter is defined on $[0,+\infty)$, depends continuously on the initial conditions and takes on values in $\Rset_+^4$.
Moreover, it is uniformly ultimately bounded.
\end{theo}


\subsection{Equilibrium points and critical gain}

We now turn to the study of the equilibrium points.
The following result shows that, for gains larger than certain critical value, the only locally asymptotically stable equilibrium is $x_{0,W}$.
Moreover, the explicit value of this critical number depends only upon the basic offspring numbers of the two populations and the ratio between their mortality rates, which are all scale-free information.

\begin{theo}
\label{th9}
If the feedback gain $K$ is such that
\begin{equation}
\label{eq12}
K > K^* := \frac{\gamma_W}{\gamma_U}\left(
\sqrt{{\cal R}_0^U}-\sqrt{{\cal R}_0^W}
\right)^2\ ,
\end{equation}
then the closed-loop system \eqref{eq20} possesses two equilibrium points, namely $x_{0,0}$ and $x_{0,W}$, and their local stability properties are not modified (i.e. $x_{0,0}$ is unstable and $x_{0,W}$ is locally stable).
\end{theo}

\begin{proof}[Proof of Theorem \ref{th9}]
\mbox{}

\noindent $\bullet$
The equilibrium points of system \eqref{eq20} are the points that verify
\begin{equation}
\label{eq200}
f(x) + KBe\t x = 0\ .
\end{equation}
Clearly, the points $x_{0,0}$ and $x_{0,W}$ are still equilibria of system \eqref{eq20}, as in these points $e\t x = L_U = 0$; and there are no other equilibria with $L_U=0$.
In fact, from the third and fourth equations of \eqref{eq200}, one should obtain
\begin{equation*}
0 = \gamma_W {\cal R}_0^W  A_W - (1+ L_W) L_W
=  (L_W^*- L_W) L_W\ ,
\end{equation*}
where the second identity follows from the definition of $L_W^*$ in \eqref{eq21}.
Let us show that there are no other equilibria than $x_{0,0}$ and $x_{0,W}$.

At any equilibrium point such that $L_U\neq 0$, \eqref{eq200} yields
\begin{equation}
\label{eq201}
0 = \gamma_W {\cal R}_0^W  A_W - (1+ L_W+ L_U) L_W + KL_U
=  ({\cal R}_0^W   - 1- L_W- L_U) L_W + KL_U
\end{equation}
and thus $L_W\neq 0$.
At such equilibrium point, one should have
\begin{equation*}
{\cal R}_0^U \frac{A_U}{A_U+A_W}
= {\cal R}_0^W+K \frac{L_ U}{L_W}
= {\cal R}_0^W+\frac{\gamma_ U}{\gamma_W}K \frac{A_ U}{A_W}
\end{equation*}
Defining the unknown quantity
\begin{equation*}
\theta := \frac{A_ U}{A_W}\ ,
\end{equation*}
the latter should fulfill
\begin{equation}
\label{eq28}
{\cal R}_0^U \frac{\theta}{1+\theta} = {\cal R}_0^W + \frac{\gamma_ U}{\gamma_W}K \theta\ ,
\end{equation}
that is
\begin{equation*}
\frac{\gamma_ U}{\gamma_W}K \theta^2
+\left(
\frac{\gamma_ U}{\gamma_W}K + {\cal R}_0^W -{\cal R}_0^U
\right)\theta
+{\cal R}_0^W = 0
\end{equation*}
The roots of this equation are given by
\begin{equation*}
\frac{\gamma_W}{2\gamma_ UK}
\left(
-\left(
\frac{\gamma_ U}{\gamma_W}K + {\cal R}_0^W -{\cal R}_0^U
\right)
\pm\sqrt{\left(
\frac{\gamma_ U}{\gamma_W}K + {\cal R}_0^W -{\cal R}_0^U
\right)^2-4{\cal R}_0^W\frac{\gamma_ U}{\gamma_W}K}
\right)
\end{equation*}

For positive values of $K$, there exist real nonnegative solutions to this equation if, and only if,
\begin{equation*}
\frac{\gamma_ U}{\gamma_W}K + {\cal R}_0^W -{\cal R}_0^U \leq 0
\qquad \text{ and } \qquad
\left(
\frac{\gamma_ U}{\gamma_W}K + {\cal R}_0^W -{\cal R}_0^U
\right)^2\geq 4{\cal R}_0^W\frac{\gamma_ U}{\gamma_W}K\ ,
\end{equation*}
that is if and only if
\begin{equation*}
\frac{\gamma_ U}{\gamma_W}K + {\cal R}_0^W -{\cal R}_0^U \leq -2\sqrt{{\cal R}_0^W\frac{\gamma_U}{\gamma_W}K}\ .
\end{equation*}
This is equivalent to
\begin{equation*}
\left(
\sqrt{\frac{\gamma_ U}{\gamma_W}K} + \sqrt{{\cal R}_0^W}
\right)^2 -{\cal R}_0^U \leq 0
\end{equation*}
or again
\begin{equation*}
\label{eq52}
K \leq \frac{\gamma_W}{\gamma_U}\left(
\sqrt{{\cal R}_0^U}-\sqrt{{\cal R}_0^W}
\right)^2
= K^*\ .
\end{equation*}
As hypothesis \eqref{eq12} is incompatible with the previous inequality, we deduce that system \eqref{eq20} possess only two equilibrium points.

\noindent $\bullet$
 We now study the local stability properties of the latter, by applying adequate modifications to the gradient exhibited in \eqref{eq50} and used in Section \ref{se232} to study the stability of the uncontrolled model equilibria.
 In fact, one just has to add to $\nabla f(x_{0,W})$ the term
 \begin{equation*}
 \begin{pmatrix}
 0 & 0 & 0 & 0\\
 0 & 0 & 0 & 0\\
 K & 0 & 0 & 0\\
 0 & 0 & 0 & 0
 \end{pmatrix}\ .
 \end{equation*}
It is clear, due to the form of this additional term, that the characteristic polynomial of the system obtained from linearizing \eqref{eq20} at $x_{0,W}$ is {\em affine} with respect to $K$, and that, for $K=0,$ it coincides with the characteristic polynomial of the linearization of \eqref{eq10}.
 
 Developing the  determinant $\det (\lambda I -\nabla f (x_{0,W}) -KBe\t)$ (see \eqref{eq26}), the additional term is equal, for the feedback control law defined in \eqref{eq01}, to
 \begin{equation*}
 \label{eq801}
 - 
 K \begin{vmatrix}
0 & 0 & 0\\
\lambda+\gamma_U & 0 & 0\\
0 & -1 & \lambda+\gamma_W
\end{vmatrix}
=0 \ .
 \end{equation*}
Therefore the local behavior is not perturbed, and the asymptotic stability of the equilibrium $x_{0,W}$ is conserved when the control term $KBe\t x$ is added.
This achieves the proof of Theorem \ref{th9}.
\end{proof}

 \begin{rema}
 It can be checked from the latter proof that the previous result is not true when the effect of cytoplasmic incompatibility is absent.
 The latter is materialized by the term $\frac{A_U}{A_U+A_W}$ present in the first line in equation \eqref{eq20}.
 When replacing this term by 1, \eqref{eq28} is replaced by
\begin{equation*}
{\cal R}_0^U = {\cal R}_0^W + \frac{\gamma_ U}{\gamma_W}K \theta\ ,
\end{equation*}
 which possesses the positive solution
 \begin{equation*}
 \theta = 
 \frac{\gamma_W}{\gamma_ UK}\left(
 {\cal R}_0^U- {\cal R}_0^W
 \right)\ ,
 \end{equation*}
leading to a coexistence equilibrium solution, in addition to $x_{0,0}$ and $x_{0,W}$.
 \end{rema}

\subsection{Global stability issues}

We now turn to the most innovative part of this paper, namely the global behavior of the closed-loop system \eqref{eq20}.
The result  we establish here shows that the introduction of infected larvae according to the proportional feedback law \eqref{eq01} yields conclusive infestation when the gain is larger than the critical value.
More precisely, we have the following convergence result.

\begin{theo}
\label{th12}
If $K>K^*$, all trajectories of system \eqref{eq20} issuing from a point in $\Rset_+^4$ distinct from $x_{0,0}$ converge towards the complete infestation equilibrium $x_{0,W}$.
\end{theo}

Notice that strictly speaking, Theorem \ref{th12} is an {\em almost global{ } convergence} result: convergence towards the complete infestation equilibrium is ensured, except for a zero measure set of initial conditions.
However, in the present case, this set is reduced to the unstable equilibrium.

Two attempts  to prove Theorem \ref{th12} are rapidly presented in Sections \ref{se331} and \ref{se332}.
The main interest is to show how these quite natural approaches fail to provide information on the asymptotic behavior and, therefore, that a new approach is needed.
Next a third conclusive method is exposed in Section \ref{se333}, where Theorem \ref{th12} is finally proved.

\subsubsection{Global stability of a singularly perturbed system, by LaSalle's invariance principle}
\label{se331}

We present now a first attempt, based on Lyapunov techniques.
Consider the simpler system
\begin{subequations}
\label{eq78}
\begin{gather}
\label{eq78a}
\dot L_U = {\cal R}_0^U \frac{\gamma_WL_U}{\gamma_WL_U+ \gamma_UL_W} L_U - (1+ L_W+ L_U) L_U \\
\label{eq78b}
\dot L_W =  {\cal R}_0^W  L_W - (1+ L_W+ L_U) L_W + KL_U
\end{gather}
\end{subequations}
As can  be easily verified, system \eqref{eq78} is deduced from \eqref{SysK} by applying singular perturbation, formally putting $0=L_U -\gamma_U  A_U$, $0 = L_W -\gamma_W  A_W$.
In other words, we assume here that \eqref{SysKb} and \eqref{SysKd} are fast dynamics, while \eqref{SysKa} and \eqref{SysKc} are comparatively much slower.

The proof of well-posedness and positiveness of system \eqref{eq78} presents no difficulties, one states directly the asymptotic properties of this system.

\begin{theo}
Assume $K>K^*$.
Then system \eqref{eq78} possesses two equilibria, namely $(0,0)$ and $(0,L_W^*)$.
The former one is unstable, while the latter one is locally asymptotically stable.
Last, all trajectories of system \eqref{eq78} converges towards $(0,L_W^*)$, except the unstable equilibrium $(0,0)$ itself.
\end{theo}

\begin{proof}
The first two assertions, proved in the same way than the similar properties were established for system \eqref{eq10}, are not detailed.
The third point is proved by considering the following candidate Lyapunov function:
\begin{equation*}
V(L_U,L_W) := \frac{L_U}{L_U+L_W}\ ,
\end{equation*}
defined in the invariant set $\Rset_+^2\setminus\{(0,0)\}$.
Letting $\dot V(L_U,L_W)$ denote the value of the derivative with respect to time of $V(L_U(t),L_W(t))$ along the trajectories of \eqref{eq78}, we have
\begin{eqnarray*}
\dot V
& = &
\nonumber
\frac{\dot L_UL_W-\dot L_WL_U}{(L_U+L_W)^2}\\
& = &
\nonumber
\frac{1}{(L_U+L_W)^2}
\left(
{\cal R}_0^U \frac{\gamma_WL_U}{\gamma_WL_U+ \gamma_UL_W}L_UL_W
-{\cal R}_0^WL_UL_W+KL_U^2
\right)\\
& = &
- \frac{L_U}{(L_U+L_W)^2(\gamma_WL_U+ \gamma_UL_W)}
\left(
{\cal R}_0^W\gamma_UL_W^2
\right.\\
& &
\hspace{2.9cm}
+ \left. \left(
{\cal R}_0^W\gamma_W+K\gamma_U-{\cal R}_0^U\gamma_W
\right)L_UL_W
+K\gamma_W L_U^2
\right)\ .
\end{eqnarray*}
Now,
\begin{eqnarray*}
\nonumber
\lefteqn{{\cal R}_0^W\gamma_UL_W^2+\left(
{\cal R}_0^W\gamma_W+K\gamma_U-{\cal R}_0^U\gamma_W
\right)L_UL_W
+K\gamma_W L_U^2}\\
& = &
\nonumber
\left(
\sqrt{{\cal R}_0^W\gamma_U} L_W - \sqrt{K\gamma_W}L_U
\right)^2\\
& &
+\left(
2\sqrt{K{\cal R}_0^W\gamma_U\gamma_W}+{\cal R}_0^W\gamma_W+K\gamma_U-{\cal R}_0^U\gamma_W
\right)L_UL_W\ .
\end{eqnarray*}
One verifies that the following factorization holds:
\begin{eqnarray}
\lefteqn{\gamma_U K
+ 2\sqrt{{\cal R}_0^W\gamma_U\gamma_W} \sqrt{K}+{\cal R}_0^W\gamma_W-{\cal R}_0^U\gamma_W}
\nonumber\\
& = &
\nonumber
\gamma_U \left(
\sqrt{K} + \frac{\sqrt{\gamma_W}}{\sqrt{\gamma_U}} \left(
\sqrt{{\cal R}_0^W} - \sqrt{{\cal R}_0^U}
\right)
\right)
\left(
\sqrt{K} + \frac{\sqrt{\gamma_W}}{\sqrt{\gamma_U}}  \left(
\sqrt{{\cal R}_0^W} + \sqrt{{\cal R}_0^U}
\right)
\right)\\
& = &
\label{eq67}
\gamma_U \left(
\sqrt{K} -\sqrt{K^*}
\right)
\left(
\sqrt{K} + \frac{\sqrt{\gamma_W}}{\sqrt{\gamma_U}}  \left(
\sqrt{{\cal R}_0^W} + \sqrt{{\cal R}_0^U}
\right)
\right)\ ,
\end{eqnarray}
where $K^*$ is the critical gain given in \eqref{eq12}.
Obviously the quantity in the right-hand side of \eqref{eq67} is positive whenever $K>K^*$.
One thus has
\begin{equation*}
\dot V(L_U,L_W) \leq 0
\end{equation*}
for any $(L_U,L_W)\in\Rset_+^2\setminus\{(0,0)\}$, with equality if and only if $L_U=0$.
LaSalle's Invariance Principle \citep{La-Salle:1976aa} is then used to conclude. 
\end{proof}

While the Lyapunov function $\frac{L_U}{L_U+L_W}$ used in the previous proof is quite appealing, it has not proved possible to extend this idea to the complete controlled system \eqref{eq20}.

\subsubsection{A monotone control system perspective}
\label{se332}
The second method now explored is an attempt to apply the results on {\em monotone control systems}, as worked out in particular by \cite{Angeli:2003aa}, see also the works by \cite{Gouze:1988aa,Cosner:1997aa,Enciso:2006ab,Enciso:2014aa}.
The principle of this approach consists in decomposing the system under study as a {\em monotone input-output system with feedback}.
Given the fact (see Theorem \ref{th55}) that the uncontrolled system is monotone, a most natural way to do this is to write system \eqref{eq20} as
\begin{subequations}
\label{eq66}
\begin{gather}
\label{eq66a}
\dot x = f(x) +KBu,\qquad y = L_U = e\t x\\
\label{eq66b}
u = y
\end{gather}
\end{subequations}
Arguing as in Theorem \ref{th55}, one can establish that the input-to-state map $u\mapsto x$ given by \eqref{eq66a} is monotone when the state space is endowed with the ordering   $\leq_{\cK}$;
while the state-to-output map $x \mapsto y = e\t x = L_U$ is anti-monotone.
We are thus in the configuration of a so-called {\em monotone system with negative feedback}.

In such a case, the study of asymptotics of the system obtained when closing the loop by the unitary feedback \eqref{eq66b} can be done by introducing {\em static characteristics} \citep{Angeli:2003aa,Enciso:2014aa}.
By definition, when it exists, the {\em input-state characteristic} $k_X$ associates to any constant input $\bar u$ the corresponding value $k_X(\bar u)$ of the unique globally asymptotically stable equilibrium; and the {\em input-output characteristic} $k$, obtained by composing $k_X$ with the state-to-output map, associates to $\bar u$ the corresponding output value $k(\bar u)$.
Notice that for a monotone system with negative feedback (as it is the case here), the map $k$ is {\em non-increasing}.
The inspiring results demonstrated in the references above establish general conditions under which {\em existence and stability of fixed points of the input-output characteristic} $k$ permit to deduce the convergence of every (or almost every) trajectory of the closed-loop system towards a locally asymptotically stable equilibrium; and that these equilibria are in one-to-one correspondence with the stable fixed points (through the state-to-output map).

\begin{figure}[ht]
\begin{center}
\mbox{\includegraphics[scale=0.5]{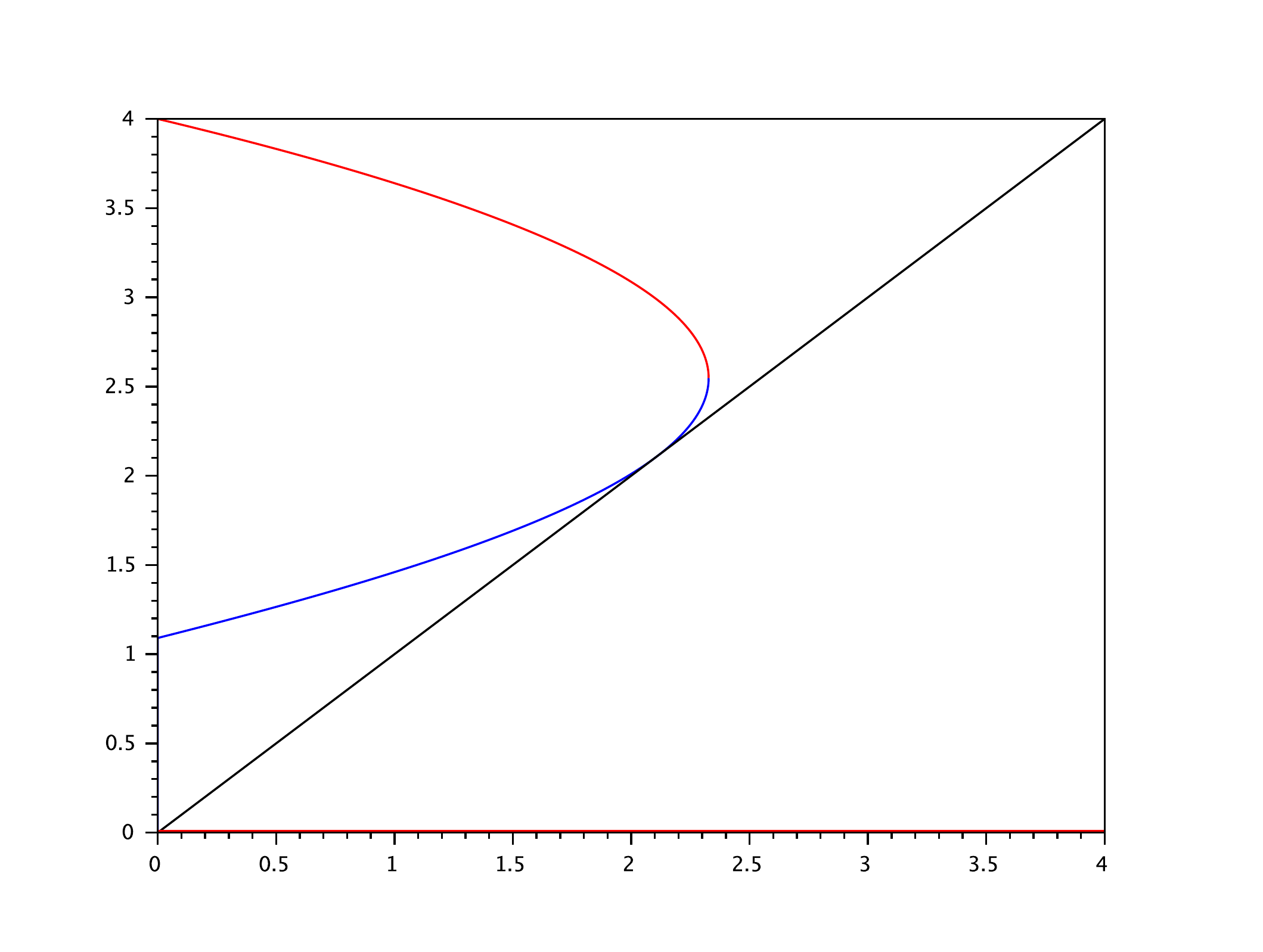}}
\caption{The multivalued input-output characteristic $\bar u \mapsto k(\bar u)$ corresponding to input-output system \eqref{eq66a} (in red) and the diagonal $\bar u \mapsto \bar u$ (in black), in the critical case $K^*=K$.
See text for explanations}
\label{fi1}
\end{center}
\end{figure}

However, things become immediately complicated in the case of system \eqref{eq66}: constant input $\bar u=0$ leads to the uncontrolled system \eqref{eq6}, which has been proved to possess {\em two} locally asymptotically stable equilibria (and two unstable ones).
In such a case, one can consider {\em multivalued input-state and input-output characteristics}, as made by \cite{Malisoff:2005aa}, or \cite{Gedeon:2009aa}.
The corresponding input-output characteristic is drawn in {\color{red} red} in Figure \ref{fi1}  (for the parameter values given in Section \ref{se4} below) and for the corresponding critical value $K^*$ of $K$.

As seen in the figure, the input-output characteristic has basically two branches (both drawn in red).
The first one merges with the horizontal axis: it corresponds to a branch of equilibria with null value of the output $y=L_U$ that departs from $x_{0,W}$ for $\bar u=0$.
The second one is a decreasing curve, defined for values of $\bar u$ ranging from zero to a value close to 2.32: it corresponds to the output value of a branch of equilibria departing from $x_{U,0}$.
The {\color{blue} blue} curve, which does {\em not} pertain to the input-output characteristic, indicates the output values of a branch of unstable equilibria originating from $x_{U,W}$ and that vanishes together with the upper curve.
The diagonal  line (that determines the fixed points of $k$) is also shown.
It is tangent to the blue curve, due to the fact that $K=K^*$ here.
For smaller values of $K$, the diagonal intersects twice the two upper branches; while for larger values of $K$ the only intersection between the input-output characteristic and the diagonal is the origin.

For $K>K^*$, the complete infestation equilibrium is therefore the only fixed point of the multivalued map $k$.
But the iterative sequences $\bar u_{k+1} = k(\bar u_k)$ do {\em not} converge systematically towards this point.
In fact, the only information that can be deduced from the results applicable to cases of multivalued input-output characteristic \citep{Malisoff:2005aa,Gedeon:2009aa}, is that all trajectories are bounded, and that the output $y = L_U$ fulfills the following inequalities:
\begin{equation*}
0 \leq \limsup_{t\to +\infty} y(t) \leq k^{\sup} \left(
\liminf_{t\to +\infty} y(t)
\right)\ .
\end{equation*}
Here $k^{\sup}$ denotes the discontinuous function whose curve is equal to the upper branch until it vanishes, and merges afterwards with the horizontal axis.

As a conclusion, the decomposition \eqref{eq66}, that seemed  a natural framework to analyze the behavior of the controlled system \eqref{eq20} immediately fails to produce a global vision of the asymptotic behavior.

\subsubsection{ Monotonicity revisited and proof of the global stability}
\label{se333}

We now concentrate without supplementary detour on the proof of Theorem \ref{th12}.
The principle consists in working on an alternative decomposition of system \eqref{eq20}, different from \eqref{eq66}.
Define first
\begin{equation*}
|z|_- := \begin{cases}
z &\text{ if } z \leq 0\\
0 &\text{ otherwise}
\end{cases}\qquad\text{ and }\qquad
|z|_+: = \begin{cases}
z &\text{ if } z \geq 0\\
0 &\text{ otherwise}
\end{cases}
\end{equation*}
Clearly, one has
\begin{equation}
\label{eq68}
z=|z|_-+|z|_+,\qquad z\in\Rset\,.
\end{equation}
We introduce the following decomposition that will show convenient in the proof of Theorem \ref{th12}.
\begin{subequations}
\label{eq013}
\begin{gather}
\label{eq013a}
\dot L_U = \gamma_U{\cal R}_0^U \frac{ A_U}{ A_U+ A_W} A_U - (1+ L_W+ L_U) L_U \\
\label{eq013b}
\dot A_U =  L_U -\gamma_U  A_U\\
\label{eq013c}
\dot L_W = \gamma_W {\cal R}_0^W  A_W - (1+ L_W) L_W
+ |K-L_W|_-L_U + K  u \\
\label{eq013d}
\dot A_W = L_W -\gamma_W  A_W\\
\label{eq013e}
y = \left|1-\frac{L_W}{K}\right|_+L_U
\end{gather}
\end{subequations}
As a matter of fact, using property \eqref{eq68}, one sees easily that the closing of the input-output link \eqref{eq013} by $u=y$ indeed yields system \eqref{eq20}.

Next, we state and prove Lemmas \ref{le11}, \ref{le13} and \ref{le12}, that will be used in the proof of the main result, Theorem \ref{th12}.

\begin{lemm}
\label{le11}
For any integrable $u$  taking on nonnegative values, the set $\Rset_+^4$ is positively invariant by \eqref{eq013}.
\end{lemm}
\begin{proof}[Proof of Lemma \ref{le11}]
The key point is that $|K-L_W|_-L_U=0$ when $0\leq L_W <K$.
Therefore, near the border of $\Rset_+^4$ where $L_W=0$, the system \eqref{eq013} behaves locally as $\dot L_W = \gamma_W {\cal R}_0^W  A_W - (1+ L_W) L_W + K  u$.
The fact that $\dot L_W\geq 0$ whenever $L_W=0$ then forbids escape from the set $\Rset_+^4$ by this side.
The same happens for the other three variables: their derivatives are nonnegative at the points where they vanish.
Hence, the trajectories can neither escape by the other sides.
This establishes the positive invariance of $\Rset_+^4$ and achieves the proof of Lemma \ref{le11}.
\end{proof}

\begin{lemm}
\label{le13}
The input-output system \eqref{eq013} is monotone with negative feedback, when the state space is endowed with the order $\geq_{\cK}$ defined in Theorem \ref{th55}.
\end{lemm}

\begin{proof}[Proof of Lemma \ref{le13}]
The right-hand sides of \eqref{eq013a}, \eqref{eq013b} and \eqref{eq013d} have been studied in Theorem \ref{th55}.
The right-hand side of \eqref{eq013c} is clearly increasing with respect to $A_W$ and $u$, and non-increasing with respect to $L_U$.
So the input-to-state map is monotone when the state space is endowed with the order defined in  Theorem \ref{th55}.

On the other hand, the state-to-output map defined by \eqref{eq013e} is non-increasing with respect to $L_W$, and non-decreasing with respect to $L_U$.
Therefore, it is anti-monotone with respect to the ordering used in the state space.
This achieves the proof of Lemma \ref{le13}.
\end{proof}

\begin{lemm}
\label{le12}
For any constant nonnegative input, system \eqref{eq013a}--\eqref{eq013d} possesses a unique LAS equilibrium.
The latter yields null value of $L_U$.

Moreover, the solution of the input-output system \eqref{eq013a}--\eqref{eq013d} converges towards the corresponding equilibrium when time goes to infinity (and in particular the output $L_U$ converges to zero), except possibly if $\bar u=0$ and $L_W(0)=0$, $A_W(0)=0$.
\end{lemm}

Lemma \ref{le12} does {\em not} allow to define in the usual way an (identically null) input-output characteristic:
when $\bar u>0$, the solution of the input-output system \eqref{eq013} converges to the LAS equilibrium for any initial condition, and is such that
\begin{equation*}
\lim_{t\to +\infty} L_U(t) =  0\ ;
\end{equation*}
but when $\bar u=0$, this property is only guaranteed if $(L_W(0),A_W(0))\neq (0,0)$.

However, Lemma \ref{le12} allows to define a weaker notion of characteristic called  {\em input-output quasi-cha\-racte\-ristic} \citep{Angeli:2004ac}.
Contrary to the stronger notion (presented in Section \ref{se332}), for any constant input value, the convergence to the asymptotic value does not have  to be global: convergence is only required to occur in a full measure set (i.e.\ a set whose complementary is of zero-measure).
This is exactly what happens in the present situation, as the set of those $x\in\Rset_+^4$ for which $L_W(0)=0$ and $A_W(0)=0$ is negligible.

As a consequence, one can in fact show, with the tools developed by \cite{Angeli:2004ac}, that $x_{0,W}$ is {\em almost-}globally attractive for system \eqref{eq20}, i.e.\ that it attracts all trajectories, except possibly those departing from certain zero-measure set.
The difficulty we now face is that this negligible set contains {\em a priori all the points of $\Rset_+^4$ such that $(L_W(0),A_W(0)) = (0,0)$, and in particular the equilibrium} $x_{U,0}$.
However,  it is quite natural to assume that, at the moment where the release begins, the system departs from a {\em Wolbachia}-free situation.
In order to ensure that these situations too are concerned by the convergence to $x_{0,W}$, more precise arguments are therefore needed.

We first prove Lemma \ref{le12}, and demonstrate in the sequel that, as announced in Theorem \ref{th12}, the convergence occurs for all trajectories, except the unique trajectory immobile at  $x_{0,0}$.

\begin{proof}[Proof of Lemma \ref{le12}]
One first studies the equilibria of system \eqref{eq013a}--\eqref{eq013d}, for constant inputs $u\ :\ t\mapsto u(t)\equiv \bar u$, for $\bar u\in\Rset_+$.
Clearly, the set of these equilibria is the union of two sets: the set of equilibria of \eqref{eq013a}+\eqref{eq013b}+\eqref{eq013d} and
\begin{subequations}
\label{eq014}
\begin{equation}
\label{eq014a}
\dot L_W = \gamma_W {\cal R}_0^W  A_W - (1+ L_W) L_W + K\bar u
\end{equation}
such that $K-L_W\geq 0$; and the set of equilibria of \eqref{eq013a}+\eqref{eq013b}+\eqref{eq013d} and
\begin{equation}
\label{eq014b}
\dot L_W = \gamma_W {\cal R}_0^W  A_W - (1+ L_W+L_U) L_W + KL_U + K\bar u
\end{equation}
\end{subequations}
such that $K-L_W\leq 0$.

Consider first the system \eqref{eq013a}+\eqref{eq013b}+\eqref{eq014a}+\eqref{eq013d}.
As can be seen, the latter is decoupled, since $L_U$ is not anymore present in the right-hand side of \eqref{eq014a}.
One shows without difficulty that there exists a unique equilibrium in $\Rset_+^4$, which is LAS and characterized by
\begin{equation}
\label{eq015}
L_W = \frac{1}{2} \left(
{\cal R}_0^W-1+\sqrt{({\cal R}_0^W-1)^2+4K\bar u}
\right)\ ,
\end{equation}
provided that this expression verifies $L_W\leq K$.
Another equilibrium exists, which is $x_{0,0}$ if $\bar u=0$, but which has {\em negative} value of $L_W$ if $\bar u>0$, and is therefore discarded, due to Lemma \ref{le11}.

Consider now the second case, of system \eqref{eq013a}+\eqref{eq013b}+\eqref{eq014b}+\eqref{eq013d}.
Arguing as in the proof of Theorem \ref{th9}, the only equilibria that may exist are such that $L_U=0$.
As a matter of fact, for a solution with nonzero $L_U$, a term $K\bar u$ in the right-hand side of \eqref{eq201} could be written, jointly with $KL_U$, as $K'L_U$ for some $K'\geq K>K^*$, leading therefore to $L_U=0$ and a contradiction.
Therefore, any potential equilibrium has to fulfill $L_U=0$, and the only possibility is given by \eqref{eq015} if this expression verifies $L_W\geq K$.

Putting together the two cases, one sees that:
\begin{itemize}
\item[$\star$]
there exist two equilibria, $x_{0,0}$ and $x_{0,W}$, if $\bar u =0$;
\item[$\star$]
there exists a unique equilibrium if $\bar u>0$.
\item[$\star$]
In any case, the corresponding output value is 0.
\end{itemize}

Now, for any constant input $u(t)\equiv \bar u$, system \eqref{eq013} is strongly order-preserving, just as system \eqref{eq6} was shown to be (Theorem \ref{th55}).
Then, the uniqueness of equilibrium in the case where $\bar u>0$ allows to use \cite[Theorem 2.3.1, p.\ 18]{Smith:1995aa} and to deduce that all trajectories in $\Rset_+^4$ converge to this unique equilibrium.
When $\bar u=0$, applying \cite[Theorem 2.2.1, p.\ 17]{Smith:1995aa} shows that every trajectory converges towards one of the two equilibria $x_{0,0}$ and $x_{0,W}$.
The behavior of system \eqref{eq013} in the vicinity of $x_{0,0}$ obeys the equations
\begin{equation*}
\dot L_W = \gamma_W {\cal R}_0^W  A_W - (1+ L_W) L_W,\qquad
\dot A_W = L_W -\gamma_W  A_W\ .
\end{equation*}
This system is monotone and the projection of $x_{0,W}$ attracts all trajectories, except if $L_W(0)=0$ and $A_W(0)=0$.
This achieves the proof of Lemma \ref{le12}.
\end{proof}

One is now ready to achieve the proof of Theorem \ref{th12}.
We define $y(t;x_0,u)$ the output of system \eqref{eq013a}--\eqref{eq013d} corresponding to the input signal $u$ and the initial state value $x_0$.
For any trajectory of the closed-loop system \eqref{eq20}, we will denote indifferently $u$ and $y$, in order to exploit the formalism of the input-output decomposition given in \eqref{eq013}.

First of all, recall that, due to Theorem \ref{th7}, all trajectories of \eqref{eq20} are bounded.
Therefore, for any nonnegative initial condition $x_0$, 
\begin{equation}
\label{eq73}
0
\leq
\liminf_{t\to +\infty} y(t;x_0,u)
\leq
\limsup_{t\to +\infty} y(t;x_0, u)
< +\infty
\qquad \forall x_0\in\Rset_+^4\ .
\end{equation}

Assume first
\begin{equation}
\label{eq777}
\liminf_{t\to +\infty} y(t;x_0,u) >0\ .
\end{equation}
Using monotonicity of the input-output system to compare trajectories with different inputs, one obtains from the fact that
\begin{equation*}
\forall\varepsilon>0,\exists T>0,\ t\geq T \Rightarrow
u(t) \geq \liminf_{t'\to +\infty} u(t')-\varepsilon\ ,
\end{equation*}
that
\begin{equation}
\label{eq74}
\forall\varepsilon>0,\
\limsup_{t\to +\infty} y(t;x_0, u)
\leq \limsup_{t\to +\infty} y\left(
t;x_0, \liminf_{t'\to +\infty} u(t')-\varepsilon
\right)\ .
\end{equation}
Using Lemma \ref{le12} for $\varepsilon>0$ smaller than $\displaystyle\liminf_{t'\to +\infty} y(t';x_0,u)$ yields
\begin{equation}
\label{eq75}
\limsup_{t\to +\infty} y\left(
t;x_0, \liminf_{t'\to +\infty} u(t')-\varepsilon
\right)
= 0\ .
\end{equation}
By putting together \eqref{eq73}, \eqref{eq74} and \eqref{eq75}, one gets:
\begin{equation}
\label{eq733}
0
<
\liminf_{t\to +\infty} y(t;x_0,u)
\leq
\limsup_{t\to +\infty} y(t;x_0, u)
\leq 0\ ,
\end{equation}
which is absurd.
This shows consequently that the premise \eqref{eq777} was erroneous.

We thus have
\begin{equation}
\liminf_{t\to +\infty} y(t;x_0,u) = 0
\end{equation}
for all trajectories.
Assume
\begin{equation}
\label{eq779}
(L_W(0),A_W(0))\neq (0,0)\ .
\end{equation}
 As above, one can deduce that
\begin{equation*}
0\leq \liminf_{t\to +\infty} y(t;x_0, u)
\leq \limsup_{t\to +\infty} y(t;x_0, u)
\leq \limsup_{t\to +\infty} y\left(
t;x_0, \liminf_{t'\to +\infty} u(t')
\right) = 0\ ,
\end{equation*}
and therefore that
\begin{equation}
\label{eq778}
\lim_{t\to +\infty} y(t;x_0,u) =0\ .
\end{equation}

Now  Lemma \ref{le12} permits to deduce from \eqref{eq778} and  \eqref{eq779} that
\begin{equation}
\label{eq780}
\lim_{t\to +\infty} x(t;x_0,u) =x_{0,W}\ .
\end{equation}

On the other hand, if \eqref{eq779} is false but $(L_U(0),A_U(0))\neq (0,0)$, then  it is easy to show that $(L_W(t),A_W(t))$ $\neq$ $(0,0)$ for some $t>0$ (and indeed for {\em any} $t>0$).
As a matter of fact, due to the presence of the control term (which is continuous and initially positive), $L_W$ is certainly positive on a sufficient small punctured open neighborhood of $t=0$.
This in turn yields the same property for $A_W$, due to the linearity of its evolution.
The analysis previously conducted in the case where \eqref{eq779} is true, can therefore be applied in the present case (where \eqref{eq779} is false but $(L_U(0),A_U(0))\neq (0,0)$) from a new, positive, initial time instant.
It allows to conclude similarly that \eqref{eq778} and \eqref{eq780} hold.

As a conclusion, the convergence  to $x_{0,W}$ occurs in any case, except if $(L_W(0),A_W(0)) = (L_U(0),A_U(0))$ $=$ $(0,0)$, that is except if $x(0)=x_{0,0}$.
This achieves the proof of Theorem \ref{th12}.

\section{Numerical simulations}
\label{se4}

We present some illustrative simulations, with the following realistic values:
\begin{equation*}
\gamma_U = 0.8,\qquad
\gamma_W = 1,\qquad
{\cal R}_0^U = 5,\qquad
{\cal R}_0^W = 3
\end{equation*}
Notice that the mortality is higher for the {\em Wolbachia} infected population ($\gamma_U<\gamma_W$), and its sustainability is inferior (${\cal R}_0^U>{\cal R}_0^W$).
The critical gain value can be computed and is equal to
\begin{equation*}
K^* \simeq 0.318
\end{equation*}

\begin{figure}[h!]
\begin{center}
\mbox{\includegraphics[scale=0.5]{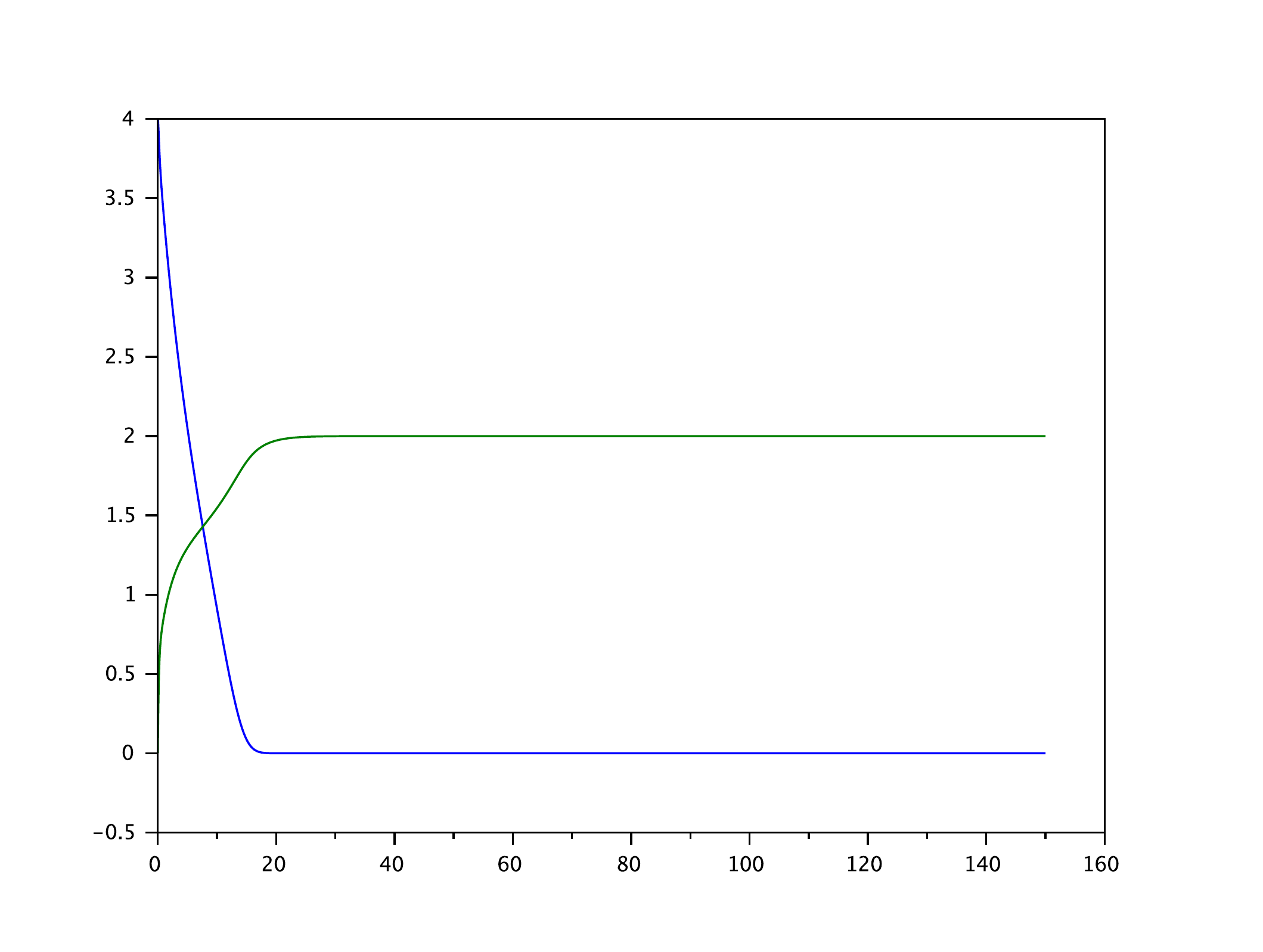}}
\caption{Evolution of {\color{blue} $L_U(t)$} and {\color{green} $L_W(t)$} as functions of time, for $K=1$}
\label{fi2}
\end{center}
\end{figure}
\begin{figure}[h!]
\begin{center}
\mbox{\includegraphics[scale=0.5]{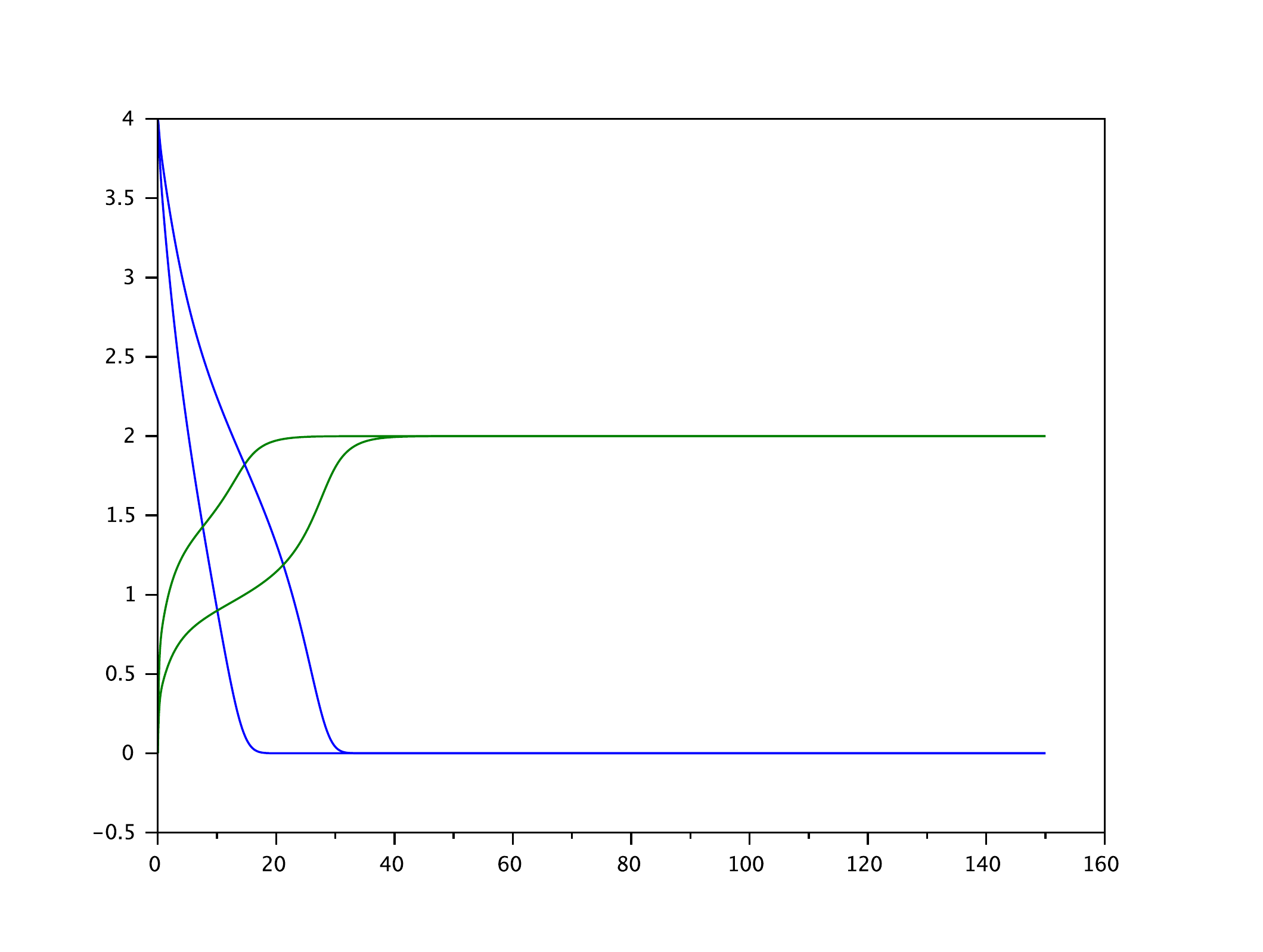}}
\caption{Evolution of {\color{blue} $L_U(t)$} and {\color{green} $L_W(t)$} as functions of time, for $K=1$ and $K=0.5$}
\label{fi3}
\end{center}
\end{figure}
\begin{figure}[h!]
\begin{center}
\mbox{\includegraphics[scale=0.5]{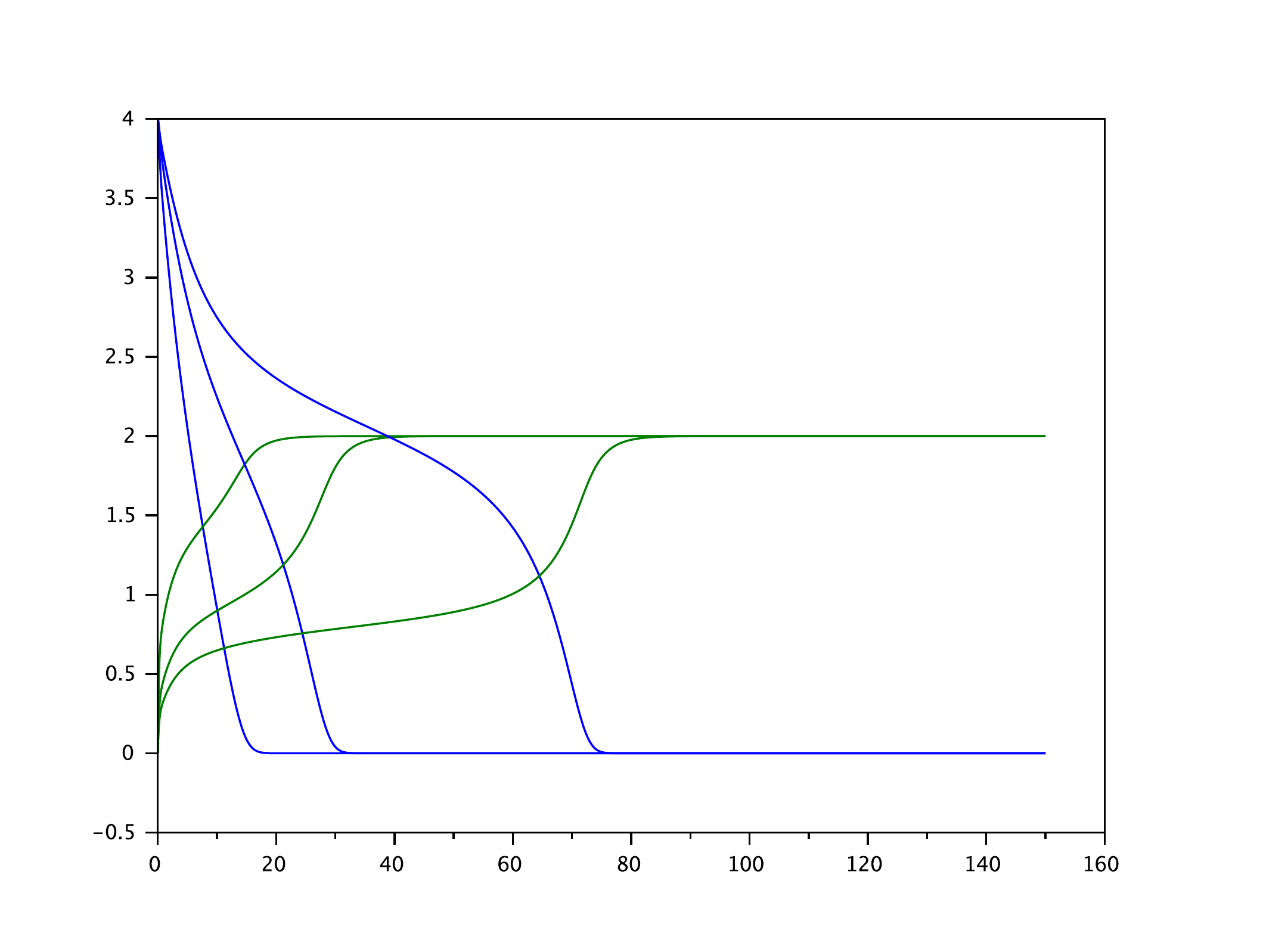}}
\caption{Evolution of {\color{blue} $L_U(t)$} and {\color{green} $L_W(t)$} as functions of time, for $K=1$, $K=0.5$ and $K=0.35$}
\label{fi4}
\end{center}
\end{figure}
\begin{figure}[h!]
\begin{center}
\mbox{\includegraphics[scale=0.5]{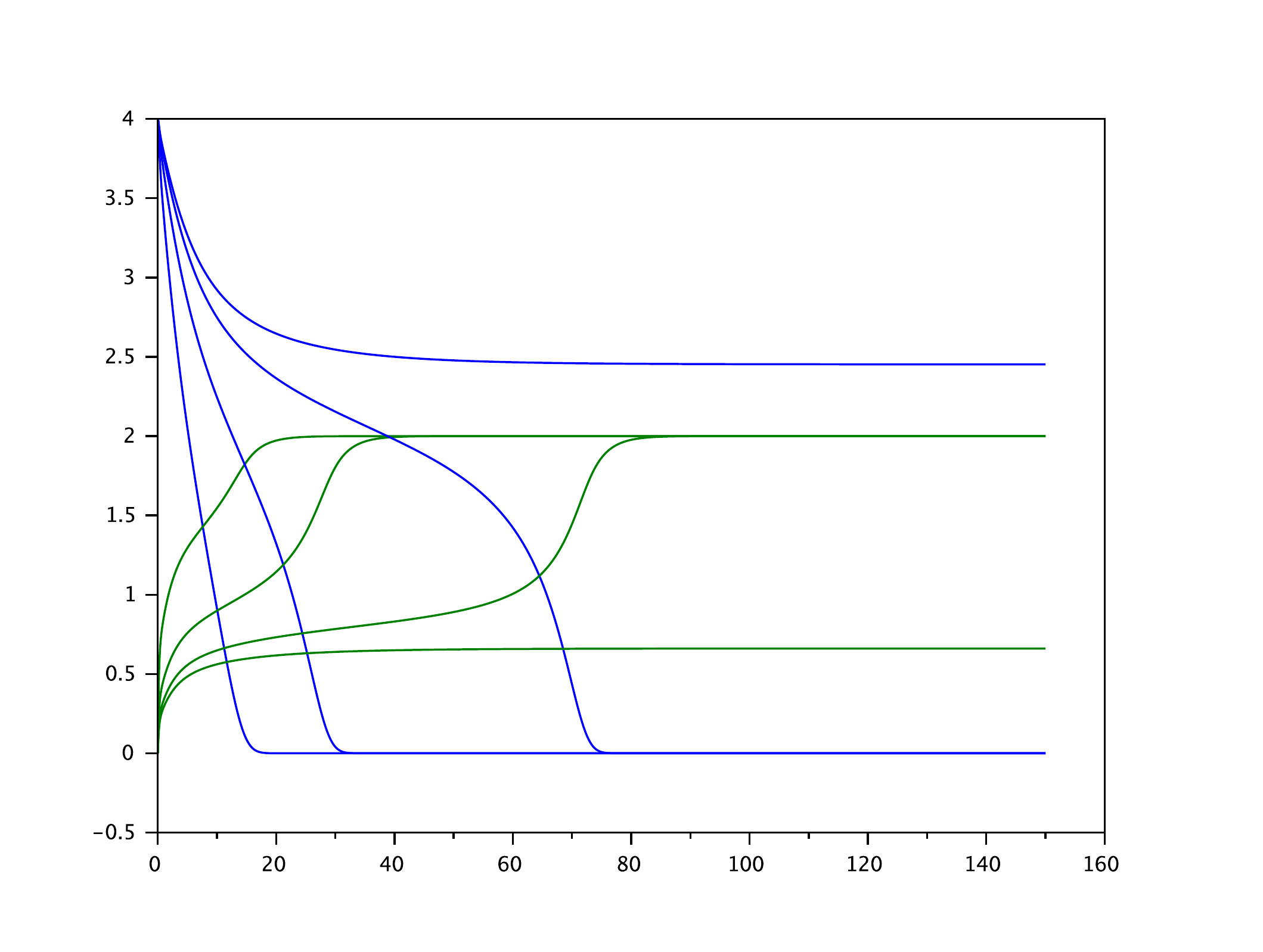}}
\caption{Evolution of {\color{blue} $L_U(t)$} and {\color{green} $L_W(t)$} as functions of time, for $K=1$, $K=0.5$, $K=0.35$ and $K=0.3$.
The last value is smaller than the critical value $K^*\simeq 0.318$, and a coexistence equilibrium appears asymptotically}
\label{fi5}
\end{center}
\end{figure}

Figures \ref{fi2} to \ref{fi5} show the evolution of the state variables $L_U$ (in {\color{blue} blue}) and $L_W$ (in {\color{green} green}) as functions of time.
The initial state is the {\em Wolbachia}-free equilibrium $x_{U,0}$, and the gain values are respectively chosen to be $1$, $0.5$, $0.35$ and $0.3$.
The last value, smaller than the critical value $K^*$, yields convergence to a coexistence equilibrium.

\section{Conclusions and further studies}
\label{se5}

We presented and analyzed a model for the infestation by bacterium {\em Wolbachia} of a population of mosquitoes --- typically one of the {\em genera Aedes} involved in the transmission of arboviroses such as yellow fever, dengue fever or chikungunya.
A method of implementation based on the introduction of a quantity of contaminated insects proportional to the size of the healthy population was proposed and shown, analytically and by simulation, to be capable to spread successfully the bacteria provided the gain is sufficiently large.
This feedback method requires continuous measurement of the population.
Its main interest with respect to the release(s) of a predefined quantity, is the reduction of the number of released mosquitoes, and thus of the treatment cost, without jeopardizing the success of the introduction of the bacteria --- something which can happen e.g.\ in case of underestimation of the initial population size.
To our knowledge, this is the first use of the control theory notion of feedback in such a context.

Among other steps leading to application, the adaptation to effective conditions has to be done. First, the model presented here has been chosen continuous in time for simplicity, but  passing to discrete-time system seems to present {\em a priori} no difficulties.
Also, the present framework assumes measurement of a larva stage of the healthy portion of the population, and as well release of larva stage of the contaminated one.
The practical conditions can be different, and the method can be adapted in consequence (leading though to similar, but different, convergence questions).
Last, issues of robustness with respect to the uncertainties of the parameters that describes the dynamics have not been tackled here.

An advantage of the present modeling framework is to open the way to comparisons with optimal policies --- for example the one that minimizes the total number of released mosquitoes, while succeeding in spreading {\em Wolbachia}.
This point will be studied in a next future.
Also, this framework provides a first basis to consider questions related to strategy improvement by mitigating several control principles, or to the complex phenomena of interaction between different vector species and different arboviruses, that may occur in the context of control of different diseases.

From a mathematical point of view, one of the difficulties of the study is that the system presents two stable equilibria, corresponding to {\em Wolbachia}-free situation and complete infestation.
While the key arguments are based on the theory of input-output monotone systems developed after \cite{Angeli:2003aa}, none of the posterior refinements to multivalued characteristics or quasi-characteristics allowed to establish formally the main convergence result, and adequate adaptation had to be achieved.
Extensions in this direction are presently studied.

\subsection*{Acknowledgements}

The first author is indebted to T.\ Gedeon for valuable discussions.
This work was done while the second author was a postdoctoral fellow at IMPA, funded by CAPES-Brazil.

\bibliographystyle{spbasic}

\bibliography{Biblio}

 \end{document}